\documentclass[notitlepage]{revtex4-1}

\usepackage[utf8]{inputenc}
\usepackage[margin=1.0in]{geometry}
\usepackage{mathrsfs}
\usepackage{hyperref}

\usepackage{amssymb,amsfonts}
\usepackage{amsthm}
\usepackage{xcolor}
\usepackage{amsmath}
\usepackage{cleveref}
\usepackage{bbold}
\usepackage{physics}
\usepackage{float}
\usepackage{natbib}
\usepackage{graphicx}
\usepackage{braket}
\usepackage{mathtools}
\usepackage{qcircuit}
\usepackage[colorinlistoftodos,prependcaption,textsize=tiny]{todonotes}
\usepackage{xargs}  

\newtheorem{theorem}{Theorem}
\newtheorem{lemma}[theorem]{Lemma}
\newtheorem*{lemma*}{Lemma}
\newtheorem{definition}[theorem]{Definition}

\newcommand{\nn}{\nonumber\\}
\newcommand{\pai}[1]{\pdv{#1}{\alpha_i}}
\newcommand{\pak}[1]{\pdv{#1}{\alpha_k}}
\newcommand{\rotil}{\rho_{\tilde{X}}}

\newcommand{\ro}[1]{\rho_{#1}}
\newcommand{\rtp}{\rho_{R\tilde{X}}}
\newcommand{\hk}{\tilde{H_k}}
\newcommand{\sigi}{\sigma^{-1}}
\newcommand{\sigps}{\sigma^{+}}

\newcommand{\infnorm}[1]{\left\|#1\right\|_{\infty}}

\newcommand{\exptn}[1]{\mathbb{E}\left[#1\right]}
\newcommand{\varnc}[1]{\text{var}\left[#1\right]}
\newcommand{\bigo}[1]{\mathcal{O}\left(#1\right)}
\newcommand{\bigotilde}[1]{\widetilde{\mathcal{O}}\left(#1\right)}
\newcommand{\floor}[1]{\left\lfloor #1 \right\rfloor}
\newcommand{\ceil}[1]{\left\lceil #1 \right\rceil}

\newcommand{\taurx}{\rho_{R \tilde{X}}}
\newcommand{\sigmarx}{\rho_R \otimes \rho_{\tilde{X}}}

\newtheorem{corollary}[theorem]{Corollary}

\DeclareMathOperator*{\argmin}{arg\,min}

\newcommandx{\change}[2][1=]{\todo[linecolor=blue,backgroundcolor=blue!25,bordercolor=blue,#1]{#2}}

\begin{document}
\title{Training quantum neural networks using the Quantum Information Bottleneck method}
\author{Ahmet Burak Çatlı}
\affiliation{Department of Physics, University of Toronto, Canada}
\author{Nathan Wiebe}
\affiliation{Department of Computer Science, University of Toronto, Toronto ON, Canada}
\affiliation{Pacific Northwest National Laboratory,  Richland WA, USA}

\begin{abstract}
We provide in this paper a concrete method for training a quantum neural network to maximize the relevant information about a property that is transmitted through the network.  This is significant because it gives an operationally well founded quantity to optimize when training autoencoders for problems where the inputs and outputs are fully quantum. We provide a rigorous algorithm for computing the value of the quantum information bottleneck quantity within error $\epsilon$ that requires $O(\log^2(1/\epsilon) + 1/\delta^2)$ queries to a purification of the input density operator if its spectrum is supported on $\{0\}~\bigcup ~[\delta,1-\delta]$ for $\delta>0$ and the kernels of the relevant density matrices are disjoint.  We further provide algorithms for estimating the derivatives of the QIB function, showing that quantum neural networks can be trained efficiently using the QIB quantity given that the number of gradient steps required is polynomial.
%This work therefore shows a way to not only compute information bottlenecks in the quantum realm, but also that algorithms can be devised that train a locally optimal channel that preserves the most amount of relevant information as it passes through a quantum neural network.
\end{abstract}
%\section{Introduction}
\maketitle

\section{Introduction}

%Quantum computing promises asymptotic speedups in certain computational tasks; however, the source and the extent of the proposed quantum advantages are not fully understood yet \cite{}. 
The field of quantum information strives to analyze the processing, transmission and storage properties of information in quantum systems. Substantial progress has been made in devising information processing techniques that can be used in a multitude of applications \cite{biamonte2017quantum,lloyd2016quantum,reiher2017elucidating,harrow2009quantum, van_Apeldoorn_2020}. Quantum machine learning has emerged as a hot topic in recent years and an increasing variety of quantum models, including quantum neural networks~\cite{amin2018quantum,kieferova2017tomography,Quantum_VAE,schuld2020circuit}, have been studied within the field to solve unsupervised as well as supervised learning problems.  Despite the advances, recent work has raised concerns surrounding the complexity of training these models because of vanishing gradients for the quantum neural networks~\cite{mcclean2018barren,wang2021noise,marrero2021entanglement}.  This means that understanding and optimizing the flow of relevant information through a quantum network, neural or otherwise, is a task of vital importance.

Unlike their classical counterparts however, the flow of information through a quantum neural network is poorly understood. One resolution to the first problem is to adapt methods that were used to study classical neural networks to the quantum setting. Classical information theory has been used to shed light on the power of the deep neural networks \cite{battiti1994using,Tishby2000TheIB,belghazi2018mutual}, and similarly quantum information theory has begun to shine light on the analogous question for quantum neural networks. 
In particular, we aim in this work to provide a method for optimizing the flow of relevant information about a concept through a quantum network and provide sufficient conditions for the training of these networks to be possible in polynomial time.
%In this work we hope to build towards a quantum informational analysis of quantum machine learning approaches to quantify the performance characteristics and expressive powers of these methods. Our wish is that this line of analysis will let us know if we stand to gain anything by executing our classical information processing tasks on a quantum computer; or if solely quantum information processing is feasible at all - for which no classical option exists.

Our instrument of choice for this task is the Information Bottleneck (IB), first proposed in \cite{Tishby2000TheIB}. The method expresses the problem of training a neural network as optimizing a variational channel with side information; the optimal channel is the one that keeps the maximum amount of information about the signal in the side channel (i.e. the label in classical machine learning terms) while keeping the output as ``simple" as possible by shedding the irrelevant information. By tracking the relevant information $I(output:label)$ and the memory information $I(input:output)$ throughout the training process, at different points of the network, it is possible to have a close look at the information dynamics throughout the training process, and distill information about how different choices of optimization method, error function and network parametrization in a theoretically grounded matter. 

Recent work has extended the notion of the information bottleneck  into the quantum domain~\cite{Datta_2019,hayashi2022efficient}; however, as yet efficient quantum algorithms for optimizing the quantum information bottleneck (QIB) have not been devised.  We address this issue and here and further argue that the QIB methodology is especially useful in the quantum settings owing to the unresolved questions surrounding what quantum properties and quantum network architectures can be best leveraged for quantum machine learning tasks.  The information bottleneck not only provides a means to quantify such questions, but our work shows how we can efficiently train such networks under reasonable assumptions on the training data and quantum models.

We provide low-cost algorithms for estimating the quantum information bottleneck quantity and also similarly inexpensive algorithms for training a quantum neural network with respect to the QIB quantity by differentiating a series approximation to the QIB objective function.  Specifically, we show that under appropriate assumptions about the spectra of the relevant density operators, the gradients can be efficiently evaluated within bounded error on a circuit-based quantum computing model. The primary method that we employ involves estimating the QIB quantity as well as well as its derivative using Fourier series techniques and estimate the various terms in the expansion using a generalized version of the swap test circuit.  We also analyze the possible ways of bounding the information quantities, using R\'enyi-$\alpha$ entropies \cite{van2014renyi,muller2013quantum} and the measured R\'enyi entropy \cite{berta2017variational}, and compare the algorithms for computing them to direct numerical differentiation of the QIB objective function. We find surprisingly that analytic expressions for the derivative of bounds on approximations to the QIB objective function actually can be more costly to evaluate than our na\"ive numerical differentiation scheme suggesting that this approach is the most efficient relative to other natural approaches to optimize the objective function or approximations thereof.

%One novel insight we provide draws from the early work in ~\cite{} that strives to connect quantum formalism with Bayesian theory of inference. The classical mutual information of two random variables $I_c(X:Y)$ takes the same form regardless if $X$ and $Y$ are causally or acausally related to each other since in both cases the joint distribution of $P(X,Y)$ lives in the the space. The quantum analogue of the mutual information is defined much in the same way by replacing probability distributions with density operators, but ~\cite{} notes that this analogy is incomplete and show that if density operators (i.e. positive, unity trace operators) are used to express the state of acausally related quantum systems, as it is usually done in the textbook quantum formalism, then one needs to use a different construct to express causally related quantum systems, such as the input and the output of a channel, if a complete analogue of Bayesian inference is desired. We make this dichotomy explicit and explore the consequences of not having the full power of Bayesian inference for learning tasks. 

The remainder of the paper is laid out as follows: Section 2 further introduces the information bottleneck method, the past developments in the classical case and the recent work in the quantum case. Sections 3 and 4 give the methods for estimating the Quantum Information Bottleneck quantity and its derivative. Section 4 provides the results on the R\'enyi-$\alpha$ and the measured R\'enyi methods, followed by the concluding remarks and the possible ways to build upon this framework in Section 5.

\section{Information Bottlenecks}
The goal of building an optimal coding scheme can be understood as a task to find an optimal trade-off between two distinct desiderata: minimizing the size of the encoded signal in a compression scheme while maximizing the relevant information retained in the compressed signal.  For example, consider a case where we wish to encode hand written digits as per the MNIST dataset~\cite{deng2012mnist}.  In this case for each image there is a corresponding label.  We wish in an encoding for such a task to maximize the information about the hidden label within the compressed data.  Naturally there are trade-offs to be made here between compression and model accuracy.  We do not simply wish to have a small compression if that compression loses the relevant information about the label, but we also do not wish to choose an encoding that trivially retains all of the relevant information by not compressing the signal.  This naturally leads to an optimization problem where the tradeoff between these two tendencies can be selected by the user through a non-negative parameter $\beta$.

One expresses the classical problem as follows: let $X$ denote a random signal with the distribution $p(x)$ and $Y$ another signal that holds auxiliary information about $X$ that we deem is important. A compression process assigns each $x$ a new codeword $\tilde{x}$, in a possibly more restricted signal space $\tilde{X}$, with probability $p(\tilde{x}|x)$. The optimal coding is expressed implicitly as a constrained optimization problem~:

\begin{equation}
    p_{opt}(\tilde{x}|x) = \argmin_p \mathcal{L}[p(\tilde{x}|x)] = I_c(\tilde{X};X) - \beta I_c(X;Y) \label{classicalop}
\end{equation}
Here $I_c(X,Y)$ is the classical mutual information which is defined as $I_c(X,Y) = D_{KL}(P_{(X,Y)}||P_X \otimes P_Y)$ where $D_{KL}$ is the Kullback-Leibler divergence. Classical mutual information can also be expressed as $I_c(X,Y) = H(X) + H(Y) - H(X,Y)$, where $H(\cdot)$ is classical marginal entropy and $H(\cdot,\cdot)$ is the classical joint entropy. 

Informally, mutual information is a symmetric measure of the information shared between $X$ and $Y$, it signifies uncertainty reduced on X by knowing Y and vice versa. Thus for independent random variables the mutual information is 0, while for perfectly correlated random variables it's equal to entropy $H(X) = H(Y)$. Thus the first term in \ref{classicalop} signifies the information retained from the original signal during the coding process while the second term signifies the information retained about a third ``relevant" signal. Thus the goal of the optimization becomes to throw away as much information as possible about the original signal while retaining the maximum amount of relevant information, with the fixed parameter $\beta$ acting as our choice of the trade-off between these competing objectives. For $\beta = 0$ the optimal solution is a fixed code that assigns every incoming signal the same codeword, effectively throwing away all of the information, while $\beta \rightarrow \infty$ the optimal solution is an identical code that retains all of the information bar redundant (i.e. perfectly correlated) degrees of freedom, to retain the maximum amount of relevant information.

\subsection{Quantum Information Bottlenecks}

The quantum generalization to this approach was first proposed in \cite{Grimsmo} in the context of lossy compression. They construct the analogue using the Rate-Distortion coding scheme and using relevant information as their distortion metric. Their results indicate that quantum channels provide an advantage over classical methods where quantum correlations are present, but they do not provide any advantage in compression in the case of classical relevant information. The work of \cite{Salek_2019} then expressed the method in terms of transmission of information through a quantum channel with side information, and derived the compression rates for an entanglement assisted classical channel. However, their results relied on the conjecture that the information bottleneck function was convex, which wasn't proven at the time. Then in \cite{Datta_2019} proved the convexity of the QIB function and provided an alternative operational meaning to the Quantum Information Bottleneck (QIB) problem.\\

%The IB method was used to study Classical Neural Networks to study their information processing capabilities \cite{}. A few methods of importing Neural Networks into quantum computing have been proposed \cite{} yet Quantum Neural Networks are still in a developmental stage. \\

%We propose to use the QIB method to study possible Quantum Neural Network architectures. In this paper, we provide an algorithm to learn an optimized quantum channel. In the second chapter, the methods necessary for the calculation and optimization of the QIB function will be described, in the third chapter we quantify the error bounds and query complexity of our algorithm. The paper ends with the fourth chapter which talks about conclusions, caveats and possible further research.

%For a quantum approach; let us have quantum information in two correlated registers X and Y, with density function $\ro{XY}$. {\color{red} Since we operationally can't retain information about the original signal as it is effectively destroyed once it goes through a channel, a reference system that holds the same correlations as X.} Let R be an auxiliary reference system that purifies X, such that $Tr(\ro{RX}^2) = 1$ and $Tr_{R}(\ro{RX}) = \ro{X}$. Let  $\Phi(\ro{X}) = \rotil$ be an arbitrary quantum channel $\Phi \in C(X, \widetilde{X})$. The objective function is defined as \cite{Salek_2019}: 

There are several ways that a natural quantum analogue to the information bottleneck quantity could be constructed.  The most natural way to generalize this concept is to replace the KL-divergence with the quantum relative entropy in the definition of the mutual information.  The quantum relative entropy is formally defined below.
\begin{definition}[Quantum Relative Entropy]
Let $A$ and $B$ be density operators acting on $\mathbb{C}^{2^n}$ for integer $n$, $A,B \in D(\mathbb{C}^{2^n})$. The quantum relative entropy is defined to be $$S(A\|B) = {\rm Tr} (A \log(A)) - {\rm Tr}(A \log(B)). $$
\end{definition}

After replacing the KL-divergence with quantum relative entropies, we have to consider how to replace the classical probability distributions with analogous quantum state operators.  Further, in order to operationalize the notion of the bottleneck we need to consider how to express the correlations correctly between the input and output distributions. We assign a Hilbert space to each locus in spacetime (also termed a quantum register), such that the classical input distribution $p(x,y)$ is replaced by $\rho_{X,Y} \in \mathcal{H}_X \otimes \mathcal{H}_Y$. The classical channel characterized by $p(\tilde{x}|x)$ is replaced by the quantum channel $\Phi: D(X) \rightarrow D(\tilde{X})$, such that $\mathcal{H}_{\tilde{X}}$ is the output Hilbert space. Due to the nature of quantum mechanics, the states $\rho_X$ and $\rho_{\tilde{X}}$ do not exist at the same time. In order to be able to estimate the correlations between the input and the output, we add another reference register with associated Hilbert space $\mathcal{H}_{R}$. The overall input $\rho_{RXY}$ is prepared such that the marginal state is equal to the marginal state of X in the original $\rho_{XY}$: ${\rm Tr}_{XY}(\rho_{RXY}) = {\rm Tr}_{Y}(\rho_{XY})$. 

Since we are operating with the knowledge of $\rho_{XY}$, construction of $R$ does not violate the no-cloning theorem. The remaining marginal distributions that are used in the paper are defined as follows: $\rho_{\tilde{X}} := {\rm Tr}_{RY}(\Phi(\rho_{RXY}))$, $\rho_{{X}} := {\rm Tr}_{RY}(\rho_{RXY})$ and $\rho_{\tilde{X}Y} := {\rm Tr}_R(\Phi(\rho_{RXY}))$.  Here the last expression follows from the fact that the reference subspace contains a copy of the marginal quantum state on $X$ and the fact that it is a copy allows us to define the joint state of inputs and outputs of the channel can be found using the quantum information contained within that subsystem.

%The Information bottleneck objective function can be understood intuitively as follows.  Consider a quantum channel that describes a communication protocol $\Phi: X \mapsto \tilde{X}$ where $X$ is the input subsystem and $\tilde{X}$ is the output subsystem. Next let us consider a subsystem $Y$ that is correlated with $\tilde{X}$ and contains the ideal output of the channel for inputs in $X$.  We want to maximize the mutual information between the subsystems $\tilde{X}$ and $Y$.  This can be thought, for example, as maximizing the classification accuracy of the channel.  Simultaneously, we want to minimize the information held in the state $\rho_{X\tilde{X}}$ relative to the two subsystems $X$ and $\tilde{X}$ independently.  The ideal circumstance is that there is no additional information stored in $\rho_{X\tilde{X}}$ than would be stored for example in $\rho_X \otimes \rho_{\tilde{X}}$.

%The quantum information bottleneck quantity involves  the quantum relative entropy, which can be understood as a divergence that measures the distinguishability of two quantum states.  For convenience, we provide a definition of the quantum relative entropy quantity below.
With these definitions in place the definition of the quantum information bottleneck objective function, defined originally in~\cite{Salek_2019}, is then provided below. 

\begin{definition}[QIB Objective Function]
The quantum information bottleneck objective function for a distribution $\rho_{RXY}$ and channel $\Phi$ is denoted $\mathcal{L}[\Phi]$ which is defined to be
$$
    \mathcal{L}[\Phi] :=  \beta I(R;\tilde{X}) - (1-\beta) I(\tilde{X};Y) =  \beta S(\rho_{R\widetilde{X}}\| \rho_R \otimes \rho_{\widetilde{X}}) - (1-\beta) S(\rho_{\widetilde{X} Y}\| \rho_{\widetilde{X}} \otimes \rho_{{Y}} )
    \label{eq:objective_func}
$$
\end{definition}
An optimal channel for the QIB objective can be formally expressed as $\Phi_{opt} = \argmin_{\Phi \in C(X, \widetilde{X})} \mathcal{L}[\Phi]$.  Here we modify the expression slightly from the original formulation so that $\beta$ is in the interval $[0,1]$ without loss of generality as scalar multiplicative factors in $\mathcal{L}$ do not meaningfully change the optimization landscape and for practical reasons the evaluation of the QIB function is easier if we can promise that $\beta$ is bounded. The interpretation of the mutual information terms remain the same as its classical counterpart. The primary difference is that mutual information terms are now quantum instead of classical.

\subsection{The use of QIB in QML}
ML is perhaps the most natural application of QIB since it allows us to understand how the relevant information flows through a quantum network~\cite{schuld2020circuit,amin2018quantum,kieferova2017tomography,beer2020training} such as a quantum autoencoder~\cite{romero2017quantum} while also generalizing previous generative objective functions for QML~\cite{kieferova2017tomography,kieferova2021quantum}.
Here we discuss such applications both to reveal these possibilities and also to clarify how the information bottleneck quantity can be used and optimized in quantum information processing more broadly.

Let us consider a discriminative task within QML wherein a hidden property of the input quantum states is sought after.  In this case, there are three fundamental spaces $X,\widetilde{X}$ and $Y$.  Here $X$ can be thought of as a subsystem of the larger quantum state that holds the inputs to the quantum model.  The subsystem $\widetilde{X}$ holds the output of the quantum channel $\Phi$.  Note that in practice the subsystem $\widetilde{X}$ can potentially intersect that of $X$, although in practice we will typically think of this as a disjoint subsystem.  Subsystem $Y$ in this case contains the ``label data" that we wish to learn.

As a particular example, consider the case where input states of the form $\ket{j}\ket{\psi_j}\ket{y(j)}$ are fed into a classifier, the first register being a ``data tag" that provides a separate identifier for each $(\ket{\psi_j},\ket{y(j)})$ pair.  Our goal is to build a quantum channel that predicts the $\ket{y(j)}$ given a particular $\ket{\psi_j}$.  The training or test sets for this machine learning problem can then be thought of, without loss of generality, as
\begin{equation}
    \rho_{XY} = \sum_j p_j\ketbra{j}{j}\otimes\ketbra{\psi_j}{\psi_j}\otimes \ketbra{y(j)}{y(j)}
\end{equation}
If we apply the quantum channel $\Phi$, which we aim to train to learn the function $y(j)$, then (using superoperator notation) the resultant quantum state becomes
\begin{equation}
    (\Phi \otimes 1_Y) \rho_{XY} = \sum_j p_j\ketbra{j}{j}\otimes\ketbra{\widetilde{\psi}_j}{\widetilde{\psi}_j}\otimes \ketbra{y(j)}{y(j)} =:\rho_{\widetilde{X}Y}
\end{equation}
We can, for all cases with a finite number of training quantum state vectors, take $p_j = 1/N_{\rm train}$, to describe a uniform distribution over all the $\ket{\Psi_j}$ (which need not be orthogonal in this setting).  The density operators $\rho_{\widetilde{X}}$ and $\rho_Y$ can then be found by taking partial traces of this distribution
\begin{equation}
    \rho_{\widetilde{X}} = {\rm Tr_Y}(\rho_{\widetilde{X}Y}), \qquad \rho_{Y} = {\rm Tr_{\widetilde{X}}}(\rho_{\widetilde{X}Y}).
\end{equation}
Thus all the quantities that are needed to compute the quantity $S(\rho_{\widetilde{X} Y}\| \rho_{\widetilde{X}} \otimes \rho_{{Y}} )$ can be derived from this state and all such distributions are natural objects that one would expect to see from the joint distributions that result from a machine learning algorithm.

The mutual information between the input and the output is trickier to operationalize due to the fact that $\rho_{X\tilde{X}}$ is not a well defined quantum state. In order to rigorously assign a value to this quantity and operationalize its estimation we utilize the conditional states framework proposed in ~\cite{spekkens2013bayesian}. In this framework, each channel $\Phi \in C(X,\tilde{X})$ is assigned a quantum conditional operator $\varrho_{\tilde{X}|X}$ defined as
\begin{equation}
    \varrho_{\tilde{X}|X} := (\Phi_{X'\rightarrow\tilde{X}} \otimes 1_X) \left(\sum_{i,j}\ketbra{i}{j}_{X'}\otimes \ketbra{j}{i}_X  \right)
\end{equation}

It can be seen that this is equivalent to the Jamiolkowski representation of the channel. Here $X'$ is an auxiliary register isomorphic to $X$ and $\{\ket{i}\}_i$ is a basis for $\mathcal{H}_X$. Notice that this operator is not positive, but is positive partial transpose (PPT) with respect to the bi-partition between $X'$ and $X$. The joint causal quantum state $\varrho_{\tilde{X}X}$ associated with this conditional operator with input $\rho_X$ is given by:
\begin{align}
    \varrho_{\tilde{X}X} &:= (1_{X'}\otimes\rho_X^{\frac{1}{2}})\varrho_{\tilde{X}|X} (1_{X'}\otimes\rho_X^{\frac{1}{2}})\\
    &= \sum_{i,j} \Phi(\ketbra{i}{j}) \otimes \rho_X^{\frac{1}{2}} \ketbra{j}{i}_X \rho_X^{\frac{1}{2}}
\end{align}

It can easily be seen that $\varrho_{\tilde{X}X}$ gives the correct marginals for both $\rho_X$ and $\rho_{\tilde{X}}$. The partial transpose of $\varrho_{\tilde{X}X}$ is a density operator:
\begin{equation}
    \rho_{\tilde{X}X}:=\varrho_{\tilde{X}X}^{T_{X}} = \sum_{i,j} \Phi(\ketbra{i}{j}) \otimes \rho_X^{\frac{1}{2}} \ketbra{i}{j}_X \rho_X^{\frac{1}{2}} \label{eqn:acausal_joint}
\end{equation}
which can be seen by recognizing that $\rho_{\tilde{X}X}$ is a product of density matrices and the Choi representation of a channel which is positive. This state can be prepared by preparing the purification of $\rho_X$ and sending the $X$ register through the quantum channel. Considering the spectral decomposition $\rho_X=\sum_i \lambda_i \ketbra{e_i}{e_i}$ and the purification $\ket{\psi}_{RX} = \sum \sqrt{\lambda_i}\ket{e_i}_R\otimes\ket{e_i}_X$. The density operator of the state after sending the X register is
\begin{align}
    \rho_{R\Tilde{X}} &= \left(1_R \otimes \Phi_{X\rightarrow X'} \right) \left(\sum_{i,j} \sqrt{\lambda_i\lambda_j}\ketbra{e_i}{e_j}_R\otimes\ketbra{e_i}{e_j}_X\right)\nonumber\\
    &= \sum_{i,j} \sqrt{\lambda_i\lambda_j}\ketbra{e_i}{e_j}_R\otimes\Phi_{X\rightarrow X'}(\ketbra{e_i}{e_j})
\end{align}

which can be seen to be equivalent to \refeq{eqn:acausal_joint} by relabeling $R\rightarrow X$ and picking the arbitrary basis $\{\ket{i}\}_i$ to be the eigenbasis of $\rho_X$, $\{\ket{e_i}\}_i$.

The main two questions that remain involve asking whether this training objective function can be computed, and further whether it can be optimized using gradient descent, in polynomial time.

\section{Direct Computation of QIB Objective}
The QIB proposal in \cite{Salek_2019} give an implicit equation and an iterative algorithm to calculate the optimal channel. Since it is not possible to use such an approach for a parametrized, we propose directly calculating the objective function and its derivative. 

First, let us begin by addressing sufficient conditions for when the QIB objective can be computed efficiently using a quantum computer.  The central challenge behind this involves computing the various cross entropies present in the problem.  The cross entropies can be computed using a number of different approaches, but here we consider using a Fourier series expansion.  Specifically, we use use Lemma 37 of \cite{van2020quantum} with a Taylor series expansion to express the logarithm as Lemma \ref{lem:fourier}.

\begin{lemma}\label{lem:logapprox}
For any  integer $K>0$ and $x\in [1/\alpha,1]$ for $\alpha>1$ we have that $
    g(x)=\sum_{n=1}^K \frac{(-1)^{n+1} (x-1)^{n}}{n}$
satisfies $|g(x) - \log(x)|\le \epsilon$ if $K \ge \frac{\log(3/\epsilon)}{\log(\alpha/(\alpha-1))}$.
\end{lemma}
\begin{proof}
It is straight forward to see that since $\log(x)$ is infinitely differentiable on $[1/\alpha,1)$ the Taylor series expansion of $\log(x)$ on the interval $[1/\alpha,1)$ is 
\begin{equation}
    \log(x) = \sum_{n=1}^K \frac{(-1)^{n+1} (x-1)^{n}}{n} + \int_{\frac{1}{\alpha}}^1 {(-1)^{K+1} (x-1)^{K+1} }\mathrm{d}x.
\end{equation}
We then have that
\begin{equation}
R_K:=|g(x) - \log(x)| \le \frac{\left(\frac{\alpha -1}{\alpha}\right)^{K+2}}{K+2}\le  \frac{\left(\frac{\alpha -1}{\alpha}\right)^{K+2}}{3}.
\end{equation}
We then find that $R_K \le \epsilon$ if

\begin{equation}
    K > \frac{\log{\frac{3}{\epsilon}}}{\log{\frac{\alpha}{\alpha-1}}}.
\end{equation}
\end{proof}

\begin{lemma}\label{lem:fourier}
Let $f:\mathbb{R} \rightarrow \mathbb{C}$ and $\lambda,\epsilon \in (0,1)$, and $T[f](x) = \sum_{k=0}^K a_k x^k$ a polynomial such that $\abs{f(x)-T[f](x)} \leq \frac{\epsilon}{4}$ for all $x\in[-1+\lambda, 1-\lambda]$. Then there exists a $\boldsymbol{c}\in\mathbb{C}^{2M+1}$ such that

\begin{equation}
    \abs{f(x) - \sum_{m=-M}^{M} c_m \exp{i \frac{\pi}{2}mx}} \leq \epsilon
\end{equation}
for all $x\in[-1+\delta, 1-\delta]$, where $M=\max \left( 2 \ceil{\log{\frac{4 \|a \|_1}{\epsilon}}, \frac{1}{\delta}},0\right)$ and $\|c\|_1 \leq \|a\|_1$. The vector $\boldsymbol{c}$ can further be classically calculated in time $\operatorname{poly}(K,M,\log{\frac{1}{\epsilon}})$.
\end{lemma}
This lemma can be used to provide an approximation to the matrix logarithm.  We provide this approximation and prove a bound on the scaling of the value of $M$ required to attain this error in the following corollary.

\begin{corollary} \label{cor:fourierfirst}
For any density matrix $\rho$ such that the spectrum of $\rho$ is bounded within $[\lambda, 1-\lambda]$ the operator $\log(\rho)$ then for any $\epsilon>0$ there exists $M>0$ and $\mathbf{c}\in \mathbb{C}^{2M+1}$ such that on the Hilbert space excluding the kernel of $\rho$
$$
\left\| \log(\rho) - \sum_{m=-M}^M c_m \exp\left\{i\frac{\pi}{2} m \rho \right\} \right\|_\infty\le \epsilon,
$$
for a value of $M \in O\left( \log \frac{1}{\epsilon} +\frac{1}{\lambda} \right)$ and $c_m$ such that  $\|c\|_1 \in O(\log\log(1/\epsilon)- \log\log(1+\lambda))$ and within the kernel of $\rho$, $\sum_{m=-M}^M c_m \exp\left\{i\frac{\pi}{2} m \rho \right\}=0$.
\end{corollary}
\begin{proof}
Proof immediately follows from the Taylor series expansion of $\log(1+x) = \sum_{k=1}^\infty (-1)^k x^k/k$ which is absolutely convergent on $x\in [-1+\lambda,1-\lambda]$.  Specifically, the truncation error from Taylor's remainder theorem states that
\begin{equation}
    \left|\sum_{k=1}^K (-1)^k x^k/k - \sum_{k=1}^\infty (-1)^k x^k/k\right|\le  (1-\delta)^K/K.
\end{equation}
Solving the following expression
\begin{equation}
    (1-\delta)^K/K = \epsilon/4.
\end{equation}
yields the following bound for a sufficient value of $K$  
\begin{equation}
    K\in \widetilde{O}(\log(1/\epsilon)/\log(1/(1-\lambda))) =\widetilde{O}(\log(1/\epsilon)/\log(1+\lambda))
\end{equation} 

We then note that this implies that the conditions of \Cref{lem:fourier} apply.  Specifically, for a value of $M\in O(\log(1/\epsilon) + 1/\lambda)$, and defining $I+x = \rho$
\begin{equation}
     \|\log(I+x) - \sum_{m=-M}^M c'_m \exp\{i\frac{\pi}{2} (\rho - I)
m  \}\|_\infty = \|\log(I+x) - \sum_{m=-M}^M c_m \exp\{i\frac{\pi}{2} \rho 
m  \}\|_\infty \le \epsilon. 
\end{equation}
The one-norm of the coefficient vector obeys from \Cref{lem:fourier} that
\begin{equation}
    \|c\|_1 \le \sum_{k=1}^K 1/k \in O(\log(K))\subseteq O(\log\log(1/\epsilon)- \log\log(1+\lambda))
\end{equation}
Finally, the result that the function takes on a value of zero at zero follows from the fact that the polynomial expansion of Lemma~\ref{lem:fourier} can be taken to be an odd expansion.
\end{proof}

This corollary shows that we can compute the value of the logarithm of a density matrix in polynomial time provided that we are dealing with either a full rank matrix or we restrict our attention to the space orthogonal to the Kernel of $\rho$.  This restriction of the space is reasonable in settings where we compute the relative entropies needed for the QIB objective function as the fact that the density matrix has zero eigenvalue will remove the singularities in the logarithm from the von Neumann entropies.  But before we focus on computing the QIB function in practice, we need to provide a specific access model within which we can identify the value of the bottleneck quantity.  In our context let us assume that we have access to unitary oracles $U_{\rho_1}, U_{\rho_2}$ that construct purifications of the density operators that we need to compute the QIB quantity. The oracles are defined below.

\begin{definition}
Let $U_{\rho_1}$ be an oracle that prepares a quantum state $\ket{\psi_{XYX'Y'}} \in \mathcal{H}_X \otimes \mathcal{H}_Y \otimes \mathcal{H}_{X'} \otimes \mathcal{H}_{Y'}$ such that ${\rm Tr}_{X'Y'}(\ketbra{\psi_{XYX'Y'}}{\psi_{XYX'Y'}}) = \rho_{XY}$.
\begin{equation}
    U_{\rho_1} \ket{0} = \ket{\psi_{XYX'Y'}}.
\end{equation}
Let $U_{\rho_2}$ be an oracle that prepares a quantum state $\ket{\psi_{RX}} = \sum \sqrt{\lambda_i}\ket{e_i}_R\otimes\ket{e_i}_X \in \mathcal{H}_X \otimes \mathcal{H}_R$.
\begin{equation}
    U_{\rho_2} \ket{0} = \ket{\psi_{RX}}.
\end{equation}
We further assume that the oracle has a known inverse that can be implemented using at the cost of a single query to $U_{\rho}$.
\end{definition}  

%Note above the subsystem R is introduced as an operational tool; it lets us employ a purified access model while allowing us to calculate the mutual information between the input and the output of the channel with a single invocation of the oracle. All of the informational quantities can be derived from $\rho_{RXY}$. 

With this unitary access model in place, we can bound the number of queries needed in order to compute the quantum information bottleneck quantity.  We summarize the result in the following theorem.
\begin{theorem}
For positive $\lambda<1$ and $\epsilon>0$ the QIB objective function $\mathcal{L}(\Phi)$ can be computed within error $\epsilon$ with probability of greater than $2/3$ using a number of calls to $U_p$ and the channel $\Phi$ that scales as
$$
O\left( \log^2(1/\epsilon) + \frac{1}{\lambda^2} \right).
$$
provided the kernels of $\rho_{\widetilde{X}} \otimes \rho_{{Y}}$ is a subspace of the kernel of $\rho_{\widetilde{X} Y}$ and similarly the kernel of $\rho_{\widetilde{X}} \otimes \rho_{{X}}$ is a subspace of the kernel of $\rho_{\widetilde{X} Y}$ and the second smallest eigenvalues for all density matrices are at least $\lambda$.
\end{theorem}
\begin{proof}
The objective function in~\eqref{eq:objective_func}  consists of a number of different density operators.  In all cases we need to estimate a corresponding entropy or cross-entropy for each such term.  First the operation $U_\rho$ prepares the quantum state $\ket{\rho_{RXY}}$ which is a purification of the input density operator over the system, the reference system and the correlated subsystem $Y$.

In order to implement the an approximation to the matrix logarithms in~\eqref{eq:objective_func} using~\Cref{cor:fourierfirst}, we need to construct a method for building terms of the form $e^{-i\rho_{X\tilde{X}} {\alpha_k}},e^{-i\rho_{X}\otimes \rho_{\tilde{X}} {\alpha_k}},\ldots$.
Corollary 17 of \cite{low2019hamiltonian} can be utilized to implement a unitary $V_\sigma(t)$ such that
\begin{equation}
    \|e^{it\sigma} - V_\sigma(t) \|_\infty \le \epsilon,
\end{equation} 
while using $\bigo{t + \log \frac{1}{\epsilon}}$ queries to an oracle $U_\sigma$ such that ${\rm Tr}_a(U_{\sigma} \ket{0}_s \ket{0}_a)=\sigma$.  We can then construct an oracle $U_{\rho_{X\tilde X}}$ in the following way.  From a single query to $U_\rho$ we can construct a state of the form
\begin{equation}
    {\rm Tr}_{RY}(U_\rho \ket{0}) = \rho_{X}.
\end{equation}
Thus by redefining the ancillary system $a$ we can implement a unitary that prepares a purification of this state, $U_{\rho_X}$, using a single query to $U_\rho$. Note that in our notation ${\rm Tr}(\ket{\psi}) \coloneqq {\rm Tr}(\ketbra{\psi}{\psi})$.

Similarly, we can construct a unitary that prepares a purification of $\rho_{X\tilde X}$ through the use of the reference system
\begin{equation}
    {\rm Tr}_{Y}((\Phi \otimes I_Y\otimes I_R)U_{\rho} \ket{0})=\rho_{X\tilde X}.
\end{equation}
We then define this oracle to be $U_{\rho_{X\tilde X}}$ which can also be implemented using a single query to $U_\rho$.  The unitary preparing an encoding of $\rho_{XY}$, $U_{\rho_{XY}}$, can also be built using a single query to the oracle.

The final state that we need is of the form $\rho_X\otimes \rho_{\tilde{X}}$  which we can construct using two queries to $U_\rho$ via
\begin{equation}
    {\rm Tr}_{YR,Y'R'}( (\Phi_{X'} \otimes I_{XYR,Y'R'})U_\rho^{\otimes 2} \ket{0}^{\otimes 2} ) = \rho_X \otimes \rho_{\tilde{X}},
\end{equation}
where the primed subsystems refer to the reference and the output subsystems of the second subsystem that $U_\rho$ acts on.  This shows that oracles can be constructed that implement purifications of each of the reduced density matrices in~\eqref{eq:objective_func}.

Thus we can implement for any $\rho'$
\begin{equation}
    \left\|\sum_{m=-M}^M c_m \exp\{i\frac{\pi}{2} m \rho' \} -\sum_{m=-M}^M c_m V_{\rho'}\left(\frac{\pi}{2} m\right)\right\|_\infty \le \|c\|_1 \max_m \left\|\exp\{i\frac{\pi}{2} m \rho' \} -  V_{\rho'}\left(\frac{\pi}{2} m\right)\right\|_\infty\le \epsilon'
\end{equation}
Next we have that any term of the form, for density matrix $\rho''$ that an estimate $E_j$ can be formed such that
\begin{equation}
    |E_j -c_j {\rm Tr}(\rho'' V_{\rho'}(m\pi/2))| \le \epsilon''
\end{equation}
using $\tilde{O}(|c_j|/\epsilon'')$ applications of $V$.  Thus as there are $M$ terms, the total cost number of $V$ applications is
 
\begin{equation}
    \sum_{m=-M}^M |c_m| /\epsilon''
\end{equation}

Next by using the generalized swap-test circuit in the number of queries needed to estimate
using a number of queries to $U_{\rho'}$ that scales as
\begin{equation}
    O\left(M \left({M} + \log(\|c\|_1 / \epsilon')\right) \right).
\end{equation}
From \Cref{cor:fourierfirst} this is in
\begin{equation}
    \widetilde{O}\left(\log^2(1/\epsilon')+ \frac{1}{\lambda_m^2} \right).
\end{equation}

The logarithm approximation in Corollary \ref{cor:fourierfirst} lets us approximate the information bottleneck quantity by evaluating traces of the form $\Tr{\rho e^{ia\sigma}}$. Since we have purified access to density matrices in question, Corollary 17 of \cite{low2019hamiltonian} can be utilized to simulate $e^{ia\sigma}$ within $\epsilon$ error in $\bigo{a + \log \frac{1}{\epsilon}}$ queries to $U_\rho$. Using this method, $\Tr{\sum_m c_m \rho e^{i\frac{\pi m}{2}\sigma}}$ can be evaluated in $\bigo{M (M+\log \frac{1}{\epsilon})}$ queries with an error upper bounded by $\|c\|_1 \epsilon$. Observing that the objective function in \eqref{eq:objective_func} contains $\bigo{1}$ number of terms that are approximated by $\sum_m c_m\Tr{ \rho e^{i\frac{\pi m}{2}\sigma}}$, we can calculate the objective function in $\bigo{(\log\frac{1}{\epsilon_1} + \frac{1}{\lambda_m}) (\log\frac{1}{\epsilon_1} +\log \frac{1}{\epsilon_2} + \frac{1}{\lambda_m})}$ queries with an error upper bounded by $\varepsilon \leq 4(\epsilon_1 + \|c\|_1 \epsilon_2)$. Rescaling $\epsilon_1$ and $\epsilon_2$ to set the total error to $\epsilon$ yields a query complexity of $\bigo{(\log\frac{1}{\epsilon} + \frac{1}{\lambda_m})^2} \in \bigo{\log^2\frac{1}{\epsilon} + \frac{1}{\lambda_m^2}}$.
\end{proof}

\section{Gradients of QIB Objective}

For quantum machine learning applications, as well as certain applications in communication, we wish to optimize the QIB objective function.  This will allow us to maximize the flow of relevant information about a quantity through part of, or the entirety of, a quantum network.  The most natural way to optimize the QIB objective function is through the gradient of the objective function.  The aim of this section is to provide an algorithm for estimating the components of the gradient.

  As a first step towards computing the gradient, let us begin by defining for $X = X_0 \otimes \widetilde{X}$ the action of the channel to be 
\begin{equation}
    \Phi(\rho) := {\rm Tr}_{X_0}(U\rho U^\dagger)\label{eq:rhoDeriv}
\end{equation}
This holds without loss of generality from Stinespring's dilation theorem; however, for applications in QML it would nevertheless be preferable to take our channel to be unitary to make its implementation on a quantum computer easier.
It further follows that we can choose $U$ without loss of generality to be of the form
\begin{equation}
    U = U_n U_{n-1} \ldots U_2 U_1,\qquad U_i := e^{-i \alpha_i H_i},\label{eq:Udef}
\end{equation}
where $\{H_i\}_i \subset {\rm Herm}(X)$ is a set of Hermitian operators.
% without loss of generality, thanks to Stinespring's dilation theorem. Here U is a unitary operator $U \in U(X)$, parametrized as  \\
From the above definitions we can see that such channels are differentiable functions of the parameters $\alpha_i$ through the following lemma.
\begin{lemma}\label{lem:rhoDeriv}
For the unitary channel, $\Phi$, acting on $\rho \in \mathbb{C}^{2^n\times 2^n}$ given in~\eqref{eq:Udef} the derivative of $\Phi(\rho)={\rm Tr}_{X_0}(U\rho U^\dagger)$ given in~\eqref{eq:rhoDeriv} with respect to the parameter $\alpha_k$ is 

\begin{gather}
    \pak{\rotil} = {\rm Tr}_{X_0}(-i [\widetilde{H_k}, U\ro{X}U^\dagger]),
\end{gather}
where $\widetilde{H_k} = (\prod_{i=1}^{k-1} e^{-i \alpha_i H_i}) H_k(\prod_{i=k-1}^{1} e^{+i \alpha_i H_i})$.
This further implies that the derivatives of the remainder of the density operators that appear in the QIB 
\begin{align*}
\frac{\partial (\sigmarx)}{\partial \alpha_k} 
    &=-i \,\,\,\Tr_{X_0}(\ro{R} \otimes [\widetilde{H_k}, U \ro{X} U^\dagger]) \\
   \frac{\partial \taurx}{\partial \alpha_k} &= -i \,\,\,\Tr_{X_0}([\mathbb{1}_{R} \otimes \widetilde{H_k}, (\mathbb{1}_{R} \otimes U)\ro{RX}(\mathbb{1}_{R} \otimes U^\dagger)])
\end{align*}
\end{lemma}

\begin{proof}
First using the fact that the derivative operation commutes with the partial trace
\begin{align}
    \pak{\rotil}
    &= {\rm Tr}_{X_0}\left(\pak{(U \ro{X} U^\dagger)}  \right)
\end{align}
Next applying the product rule using the ansatz of~\eqref{eq:Udef}
\begin{align}
    \pak{(U \ro{X} U^\dagger)} &= \pak{} [(\prod_{i=1}^n e^{-i \alpha_i H_i})\ro{X} (\prod_{i=n}^1 e^{+i \alpha_i H_i})] \\
    &=  [(\pak{}\prod_{i=1}^n e^{-i \alpha_i H_i})\ro{X} (\prod_{i=n}^1 e^{+i \alpha_i H_i})] + [(\prod_{i=1}^n e^{-i \alpha_i H_i})\ro{X} (\pak{} \prod_{i=n}^1 e^{+i \alpha_i H_i})] \\
    &= [(\prod_{i=1}^{k} e^{-i \alpha_i H_i})(-i H_k)(\prod_{i=k+1}^{n} e^{-i \alpha_i H_i}) \ro{X} (\prod_{i=n}^1 e^{+i \alpha_i H_i})]   \nonumber \\ &\,\,\,\,\,\,\,\,\,\,\,\,+[(\prod_{i=1}^n e^{-i \alpha_i H_i})\ro{X} (\prod_{i=n}^{k} e^{+i \alpha_i H_i})(+i H_k)(\prod_{i=k+1}^{1} e^{+i \alpha_i H_i})] \\
    &= -i (\hk U \ro{X} U^\dagger - U \ro{X} U^\dagger \hk)= -i [\widetilde{H_k}, U \ro{X} U^\dagger]
\end{align}

Hence
\begin{gather}
    \label{eq:rotilder}
    \pak{\rotil} = {\rm Tr}_{X_0}\left(-i [\widetilde{H_k}, U \ro{X} U^\dagger] \right)
\end{gather}
Where $\widetilde{H_k} = (\prod_{i=n}^{k+1} e^{-i \alpha_i H_i}) H_k(\prod_{i=k+1}^{n} e^{+i \alpha_i H_i})$.
The Hermitian operators can then expanded in a suitable unitary basis such that $\mathbb{1} \otimes \widetilde{H_k} = \sum_{j=1}^M c_{j,k} V_j$. 

The remainder of the derivatives can then be found similarly,
\begin{align}
    \frac{\partial (\sigmarx)}{\partial \alpha_k} &= -i \,\,\,\Tr_{X_0}([\mathbb{1}_{'} \otimes \widetilde{H_k}, (\mathbb{1}_{R} \otimes U)\ro{R}\otimes\ro{X}(\mathbb{1}_{R} \otimes U^\dagger)])\nn
    &=-i \,\,\,\Tr_{X_0}(\ro{R} \otimes [\widetilde{H_k}, U \ro{X} U^\dagger]) \label{eq:sigder}\\
   \frac{\partial \taurx}{\partial \alpha_k} &= -i \,\,\,\Tr_{X_0}([\mathbb{1}_{R} \otimes \widetilde{H_k}, (\mathbb{1}_{R} \otimes U)\ro{RX}(\mathbb{1}_{R} \otimes U^\dagger)]) \label{eq:roder}
\end{align}
\end{proof}

Here we discuss what is arguably the most natural method for optimizing the QIB objective function for a network: computing numerical derivatives of the objective function through high-order divided difference methods.  We will see that, while low order methods are obviously impractical with respect to their scaling with respect to the error tolerance, high order methods form a scalable and practical method for evaluating the function.  Specifically, we approach the problem by approximating the logarithm function in an interval excluding the singularity with a series expansion. We then will be able to express the traces as expectation values, which can be calculated (within an arbitrarily small error and with arbitrarily small failure probability) with amplitude estimation methods.
The first step towards this goal is to show to construct a polynomial approximation, which we had demonstrated in Lemma \ref{lem:logapprox}.
%Let $f(\rho, \sigma) := \Tr{\rho \log{\sigma}}$.

%\begin{lemma*}[Restatement of Lemma 4]
%For any  integer $K>0$ and $x\in [1/\alpha,1]$ for $\alpha>1$ we have that $
%    g(x)=\sum_{n=1}^K \frac{(-1)^{n+1} (x-1)^{n}}{n}$
%satisfies $|g(x) - \log(x)|\le \epsilon$ if $K \ge \frac{\log(3/\epsilon)}{\log(\alpha/(\alpha-1))}$.
%\end{lemma*}

Now that we have provided this expression for a series expansion for a logarithm, we can move forward towards an expression for the derivative of an approximation to the logarithm.  Note that to differentiate the resultant power series term by term, we need to assert uniform convergence.  The result of Lemma~\ref{lem:logapprox} shows such a uniform convergence result and thus we can differentiate the series term by term to obtain the derivative.  We provide the following lemma as an intermediate result that gives a method for differentiating a Fourier series for the logarithm.  The result is given below.

\begin{lemma}
Let us define $\log_{K,M}(x) := \sum_{m=-M}^{M} c_m \exp{i \frac{\pi}{2}mx}$ for constants $c_m$ as per~Corollary~\ref{cor:fourierfirst}, where we left the K index in to denote that the coefficients were calculated using a $K^{\rm th}$ degree polynomial. Let $\rho$ and $\sigma$ be a differentiable density operator acting on a finite dimensional Hilbert space.  We can then calculate the derivative of the cross entropy assuming the kernel of $\sigma$ is a subspace of the kernel of $\rho$ 
$$
\frac{\partial}{\partial {\alpha_k}} \Tr{\rho \log_{K,M}\sigma} = \Tr{\frac{\partial \rho}{\partial {\alpha_k}}\log_{K,M}\sigma} + \sum_{m=-M}^{M} \frac{i \pi m c_m }{2} \mathbb{E}_S \left[\Tr{\rho e^{i s \frac{\pi}{2}m\sigma} \frac{\partial \sigma}{\partial {\alpha_k}} e^{i (1-s) \frac{\pi}{2}m\sigma}} \right].
$$
\end{lemma}
\begin{proof}
Let $g(x)=\sum_{m=0}^{K}a_m x^m$, let $\widetilde{f}(\rho,\sigma) := \Tr{\rho g(\sigma)} = \sum_{m=1}^{K+1}a_m \Tr{\rho \sigma^m}$. The derivative of $\widetilde{f}(\rho,\sigma)$ will then take the form

\begin{align}
    \frac{\partial}{\partial {\alpha_k}}\widetilde{f}(\rho,\sigma) &= \frac{\partial}{\partial {\alpha_k}} \sum_{m=1}^{K}a_m \Tr{\rho \sigma^m} \\
    &= \sum_{m=1}^{K}a_m \left[\Tr{\frac{\partial \rho}{\partial {\alpha_k} }\sigma^m} + \sum_{i=0}^{m-1} \Tr{\sigma^i \frac{\partial \sigma}{\partial {\alpha_k}} \sigma^{m-i-1}} \right]
\end{align}

\begin{align}
    &\frac{\partial}{\partial {\alpha_k}} \Tr{\rho \log_{K,M}\sigma} = \Tr{\frac{\partial \rho}{\partial {\alpha_k}}\log_{K,M}\sigma} + \sum_{m=-M}^{M} c_m \Tr{\rho \frac{\partial}{\partial {\alpha_k}}\exp{i \frac{\pi}{2}m\sigma}}\\
    &\qquad= \sum_{m=-M}^{M} c_m \Tr{\frac{\partial \rho}{\partial {\alpha_k}} e^{i \frac{\pi}{2}m\sigma}} + \sum_{m=-M}^{M} \frac{i \pi m c_m }{2}  \int_0^1 \Tr{\rho e^{i s \frac{\pi}{2}m\sigma} \frac{\partial \sigma}{\partial {\alpha_k}} e^{i (1-s) \frac{\pi}{2}m\sigma}}ds
\end{align}

We can express the integral as an expectation value of a random variable S with a uniform distribution in the interval $[0,1]$:

\begin{equation}
    \int_0^1 \Tr{\rho e^{i s \frac{\pi}{2}m\sigma} \frac{\partial \sigma}{\partial {\alpha_k}} e^{i (1-s) \frac{\pi}{2}m\sigma}}ds = \mathbb{E}_S \left[\Tr{\rho e^{i s \frac{\pi}{2}m\sigma} \frac{\partial \sigma}{\partial {\alpha_k}} e^{i (1-s) \frac{\pi}{2}m\sigma}} \right]
\end{equation}

Thus the derivative is expressed as:

\begin{align}
    \frac{\partial}{\partial {\alpha_k}} \Tr{\rho \log_{K,M}\sigma} &= \Tr{\frac{\partial \rho}{\partial {\alpha_k}}\log_{K,M}\sigma} + \sum_{m=-M}^{M} \frac{i \pi m c_m }{2} \mathbb{E}_S \left[\Tr{\rho e^{i s \frac{\pi}{2}m\sigma} \frac{\partial \sigma}{\partial {\alpha_k}} e^{i (1-s) \frac{\pi}{2}m\sigma}} \right] \label{eq:41derivative}
\end{align}
Next, note that as $\log_{K,M}$ is zero in the kernel of $\sigma$ and the second term above is zero in the kernel of $\rho$.  Thus the derivative is zero in the kernel of $\sigma$ if the kernel of $\sigma$ is a subspace of the kernel of $\rho$.
\end{proof}
%\subsubsection{Error analysis}
The next step towards understanding the complexity of evaluating the gradient of the QIB function involves performing error analysis on the  derivative expression.  This is vital because it will inform us about the level of precision that we will need to ensure that the final gradient estimate is appropriately accurate.  Our approach for the derivative estimation involves several steps.  We first begin by approximating the Logarithm using a Fourier series and then differentiate the Fourier series to optimize our approximation to the objective function.  In the following let $\log_{K,M}$ be an approximation to the logarithm function such that for any $\lambda_i$ inside a compact domain that excludes $0$ we have that $|\log_{K,M}(\lambda_i) - \log(\lambda_i)|\le \epsilon$.  

In order to proceed with out error analysis we need a slightly more specific expansion for approximate logarithm that we use.  Specifically we use a further truncated Taylor series expansion of $(\frac{\arcsin{x}}{\pi/2})^k$ such that for coefficients $b_l^{(k)}$, 
\begin{equation}
\left(\frac{\arcsin{x}}{\pi/2}\right)^k = \sum_{l=0}^\infty b_l^{(k)} x^l,
\end{equation}
 for $x \in (-1,1)$.  We then truncate the Taylor series expansion to order $L$ and attain the following approximation to the logarithm
\begin{equation}
\log_{KLM}\sigma :=  \sum_{k=1}^K \frac{(-1)^k}{k} \sum_{l=1}^L b_l^{(k)} (\frac{i}{2})^l \sum_{m=\ceil{l/2}-M}^{\floor{l/2}+M} (-1)^m e^{i(2m-l)\frac{\sigma\pi}{2}}.
\end{equation}
We then find that the error in the approximation to the derivative of this expression  with respect to the variational parameters is 

%\NW{Hey Burak can you check this stuff, the $\log^s$ term I think is the expression for the approximate derivative of the log but there's a bunch of notational issues below that make it hard for me to figure out exactly what you mean.  I think we want to define a term that approximates the derivative via sampling and want the difference between the true derivative and the sampled derivative.  Regardless, we need to define the $\log^s$ notation since I really think what we want is an approximation to~\eqref{eq:41derivative} for a finite number of $s$ samples.}
\begin{align}
    E &= |\partial_{\alpha_k} \Tr{\rho \log \sigma} - \partial_{\alpha_k} \Tr{\rho \log_{KLM} \sigma}|\\
    &= |\partial_{\alpha_k} (\Tr{\rho \log \sigma} - \Tr{\rho \log_{KLM} \sigma})| \\
    &= |\Tr{(\partial_{\alpha_k} \rho)(\log \sigma - \log_{KLM} \sigma)} + \Tr{\rho \partial_{\alpha_k} (\log \sigma - \log_{KLM} \sigma)}|\\
    &\leq |\Tr{(\partial_{\alpha_k} \rho)(\log \sigma - \log_{KLM} \sigma)}| + |\Tr{\rho \partial_{\alpha_k} (\log \sigma - \log_{KLM} \sigma)}| \label{eq:Ebd}
\end{align}
Focusing on the first term, we have that
\begin{align}
    |\Tr{(\partial_{\alpha_k} \rho)(\log \sigma - \log_{K,M} \sigma)}| &= |\Tr{(\partial_{\alpha_k} \rho)(\log \sigma - \log_{K,M} \sigma)}| \\
    %&= |\Tr{(\partial_{\alpha_k} \rho)\sum_i(\log \lambda_i - \log_{K,M}^s \lambda_i)}|\\
    &= \left|\Tr{(\partial_{\alpha_k} \rho)\sum_i(\log \lambda_i - \log_{KLM} \lambda_i)\ket{e_i}\bra{e_i}}\right|\\
    &\leq |\sum_i \bra{e_i}\partial_{\alpha_k} \rho\ket{e_i} \epsilon|\\
    &= \epsilon \abs{\Tr{\partial_{\alpha_k} \rho}}\\
    &= 0
\end{align}
Here the last equality follows from the fact that the trace of a commutator is zero and Lemma~\ref{lem:rhoDeriv}. The remaining term in~\eqref{eq:Ebd} can be bound using the following result, which is proven in Appendix~\ref{sec:errorproof}.
%\NW{$\epsilon$ isn't defined yet also this may be over-used so we might want a different symbol associated with this.} 

%Letting $\sigma_i(\rho)$ denote the $i^{\rm th}$ singular value of $\rho$ we have from the von Neumann trace inequality and the triangle inequality
%\begin{align}
%    |\Tr{\rho \partial_{\alpha_k} (\log \sigma - \log_{K,M}^s \sigma)}| &= |\Tr{\rho \partial_{\alpha_k} (\log \sigma - \log_{K,M}^s \sigma)}| \nn
%    &\leq \sum_i \sigma_i(\rho)\infnorm{\partial_{\alpha_k} (\log \sigma - \log_{K,M}^s \sigma)} \nn
%    &= \sum_i \sigma_i(\rho)\infnorm{\partial_{\alpha_k} (\log \sigma - \log_{K,M} \sigma + \log_{K,M} \sigma - \log_{K,M}^s \sigma)}\nn
%    &\leq \sum_i \sigma_i(\rho)(\infnorm{\partial_{\alpha_k} (\log \sigma - \log_{K,M} \sigma)} + \infnorm{\partial_{\alpha_k}(\log_{K,M} \sigma - \log_{K,M}^s \sigma)})\nonumber\\
%    &= \infnorm{\partial_{\alpha_k} (\log \sigma - \log_{K,M} %\sigma)} + \infnorm{\partial_{\alpha_k}(\log_{K,M} \sigma - \log_{K,M}^s \sigma)}\label{eq:E2}
%\end{align}

%Here the error is comprised of two terms, the first one is the error from approximating the logarithm and the second one the sampling error. Each part can be bounded separately.  

\begin{lemma}\label{lem:logapprox2}
Let $K,L,M$ be non-negative integers and let $\sigma$ be a density matrix that depends on the parameter $\alpha_k$
For values of $K$ that satisfy the assumptions of \Cref{lem:logapprox}, the error incurred by approximating the derivatives of the logarithm constrained on the subspace orthogonal to the kernel of $\sigma$ is bounded by

$$
    \infnorm{\partial_{\alpha_k} \left(\log \sigma - \log_{KLM} \sigma\right)} \leq \infnorm{\partial_{\alpha_k} \sigma}\left(\frac{1+e^2}{e} H_K\frac{L(1-\lambda_{min}^2)^{L+1}}{\lambda_{min}^3 (2-\lambda_{min}^2)^{1.5}} + 2\pi H_K L e^{-\frac{2M^2}{L}} +  \frac{(1-\lambda_{min})^K}{\lambda_{min}}\right) $$ 
    and further, there exists an algorithm that yields an estimate $E$ of ${\rm Tr}( \rho \partial_{{\alpha_k}}\log_{KLM} \sigma)$ that has zero mean error and standard deviation
    $$
    \sqrt{\mathbb{E}({\rm Tr}^2( \rho \partial_{{\alpha_k}}\log_{KLM} \sigma)) - \mathbb{E}({\rm Tr}^2( \rho \partial_{{\alpha_k}}\log \sigma))} \leq \pi M H_K \frac{\infnorm{\partial_{\alpha_k} \sigma}}{\sqrt{n}}
    $$
    Here n is the number of samples used to estimate the expectation value and $H_K$ is the $K^{\rm th}$ harmonic number, which is in $O(\log(K))$.
\end{lemma}

 In order to ensure that the total error adds up to $\epsilon$ we set $\infnorm{\partial_{\alpha_k} \left(\log \sigma - \log_{KLM} \sigma\right)} \le \epsilon/2$.  This implies that it suffices to set the errors in~Lemma~\ref{lem:logapprox} to be at most equal to $\epsilon/4$.  Then solving for approiate values of $K,L$ and $M$ to ensure these inequalities we find that it suffices to pick (see Appendix~\ref{sec:errorproof} for more details). 

%Dividing the sub-errors in $\infnorm{\partial_{\alpha_k} (\log \sigma - \log_{K,M} \sigma)}$ gives us the lower bounds on variables $K,L,M$ by setting $\infnorm{\partial_{\alpha_k} (\log \sigma - \log_{K,M} \sigma)} \le \epsilon/4$:

\begin{gather}
    K \geq \frac{\ln12\frac{\infnorm{\partial_{\alpha_k} \sigma}}{\lambda_m \epsilon}}{\ln\left(\frac{1}{1-\lambda_m}\right)},\label{eq:Klim}\\
    L \geq -\frac{1}{\ln\left(\frac{1}{1-\lambda_m^2}\right)} W_{-1}\left(\frac{\ln\left(\frac{1}{1-\lambda_m^2}\right)}{\frac{1+e^2}{e} \infnorm{\partial_{\alpha_k} \sigma} \|a\|_1^{(k)} \frac{1-\lambda_m^2}{\lambda_m^3 (2-\lambda_m^2)^{1.5}}} \frac{-1}{\frac{12}{\epsilon}}\right),\\
    M \geq \sqrt{\frac{L}{2} \ln \left( \frac{24\pi L \infnorm{\partial_{\alpha_k} \sigma}\|a\|_1^{(k)}}{\epsilon} \right)}.\label{eq:Mlim}
\end{gather}
Justification for these sufficient values is given in appendix~\ref{sec:errorproof}.
Here $W_{-1}(\cdot)$ is the Lambert-W function, and $\lambda_m$ is the smallest eigenvalue of $\sigma$. Note that $-W_{-1}(\frac{-1}{x})$ grows logarithmically in $1/x$. 

\begin{theorem}
    \label{thm:queryBd}
    There exists a quantum algorithm that can estimate the derivative of the QIB objective function within error $\epsilon$ and probability of failure at most $\delta$ using a number of queries that scale as
    $$
        \bigotilde{\frac{1}{\epsilon^{4}} \frac{1}{\lambda^2} \infnorm{\partial_{\alpha_k} \rotil}^{4} \log\left({\frac{1}{\delta}}\right) \log(\|b_k\|_1)  \max(\infnorm{\rho_X}, \infnorm{\rho_Y})}.
    $$
\end{theorem}
\begin{proof} 
Begin by substituting in $\sigma = \rho_X \otimes \rotil$ (or $\sigma = \rho_Y \otimes \rotil$), the infinite norm is multiplicative over the tensor product:

\begin{align}
    \infnorm{\partial_{\alpha_k} \sigma} &= \infnorm{\partial_{\alpha_k} (\rho_X \otimes \rotil)} \\
    &= \infnorm{ \rho_X \otimes \partial_{\alpha_k} \rotil},\\
    &= \infnorm{\rho_X} \infnorm{\partial_{\alpha_k} \rotil}.
\end{align}

Thus $\infnorm{\partial_{\alpha_k} \sigma} \leq \max(\infnorm{\rho_X}, \infnorm{\rho_Y}) \infnorm{\partial_{\alpha_k} \rotil}$. Setting the sampling standard deviation of the derivative $\infnorm{\partial_{\alpha_k}(\log_{K,M} \sigma - \log_{K,M}^s \sigma)}$ to be at most $\epsilon/4$ yields

$$
    n \geq \frac{16\pi^2 M^2 H_K^2 \infnorm{\partial_{\alpha_k} \rotil}^2 \max(\infnorm{\rho_X}, \infnorm{\rho_Y})^2} {\epsilon^2}.
$$

The number of queries needed to achieve $\epsilon/4$ error with failure probability less than $\delta$ is thus from the Chernoff bound in $\bigo{\frac{ M^2 H_K^2 \infnorm{\partial_{\alpha_k} \rotil}^2 \max(\infnorm{\rho_X}, \infnorm{\rho_Y})^2}{\epsilon^2}\ln{\frac{1}{\delta}}}$. Quadratic improvements can be attained using amplitude estimation, but for simplicity we ignore such optimizations here.  The number of queries to product traces in \ref{eq:41derivative} then scales like 

\begin{equation}
    \bigo{\frac{ M^3 H_K^2 \infnorm{\partial_{\alpha_k} \rotil}^2 \max(\infnorm{\rho_X}, \infnorm{\rho_Y})^2}{\epsilon^2}\ln{\frac{1}{\delta}}}.
    \label{eq:bigomultiplier}
\end{equation}

The last hurdle is evaluating terms like $\Tr{\rho e^{i s \frac{\pi}{2}m\sigma} \frac{\partial \sigma}{\partial {\alpha_k}} e^{i (1-s) \frac{\pi}{2}m\sigma}}$. This can be implemented using the techniques of \cite{low2019hamiltonian}, which states that as long as we have access to the purification of a density matrix $\rho$, the time evolution $e^{-i\rho t}$ can be implemented with $\epsilon$ error in trace distance in $\bigo{t + \log\frac{1}{\epsilon}}$ queries. Substituting in the definitions for $\sigma$ and $\rho$ we see that

\begin{align}
    \Tr{\rho e^{i s \frac{\pi}{2}m\sigma} \frac{\partial \sigma}{\partial {\alpha_k}} e^{i (1-s) \frac{\pi}{2}m\sigma}} &:= \Tr{\rtp \left( \ro{R} e^{i\frac{\pi m}{2}\ro{R}} \otimes e^{i s \frac{\pi}{2}m\rotil} (\partial_k \rotil) e^{i (1-s) \frac{\pi}{2}m\rotil} \right)}
    \label{eq:1beforefinalsum}
\end{align}

Substituting the definition of $\partial_k \rotil = Tr_{X_0}\left(-i [\widetilde{H_k}, U \ro{X} U^\dagger] \right)$ and expanding $\widetilde{H_k} = \sum_i^{N_u} b_{ik} V_i$ in a suitable unitary basis, we arrive at

\begin{align}
&\Tr{\rtp \left( \ro{R} e^{i\frac{\pi m}{2}\ro{R}} \otimes e^{i s \frac{\pi}{2}m\rotil} (\partial_k \rotil) e^{i (1-s) \frac{\pi}{2}m\rotil} \right)}\nn
    \label{eq:finalsum}
    &\qquad= -i \sum_{j=1}^{N_u} b_{jk} \Tr{\rtp \left( \ro{R} e^{i\frac{\pi m}{2}\ro{R}} \otimes e^{i s \frac{\pi}{2}m\rotil} [V_i, U \ro{X} U^\dagger] e^{i (1-s) \frac{\pi}{2}m\rotil} \right)}
\end{align}

Setting the error for a single trace in the sum in \refeq{eq:finalsum} equal to $\epsilon_H / \|b_k\|_1$, we get a total error of $\epsilon_H$ for the expression in \ref{eq:1beforefinalsum}. For the exponentials in \refeq{eq:finalsum} we need 3 Hamiltonian simulation operations which can be performed using the density matrix exponentiation routine in~\cite{low2019hamiltonian} with a query complexity $\bigo{M+\log\frac{\|b_k\|_1}{\epsilon}}$.  This holds because we assume that we have access to a unitary state preparation routine that constructs a purification of the required density operators in accordance with the assumptions of the method of~\cite{low2019hamiltonian}.  Considering there are $N_u$ summands in \eqref{eq:finalsum}, $ \bigo{N_u(M+\log\frac{\|b_k\|_1}{\epsilon})}$ queries are needed to evaluate \eqref{eq:1beforefinalsum}. Multiplying through with \eqref{eq:bigomultiplier}, we obtain a final query complexity of 

\begin{gather}
    \bigo{\frac{ M^3 H_K^2 \infnorm{\partial_{\alpha_k} \rotil}^2 \max(\infnorm{\rho_X}, \infnorm{\rho_Y})^2}{\epsilon^2} (M+\log\frac{\|b_k\|_1}{\epsilon})\ln{\frac{1}{\delta}}}\\
    = \bigotilde{\frac{1}{\epsilon^{4}} \frac{1}{\lambda^2} \infnorm{\partial_{\alpha_k} \rotil}^4 \max(\infnorm{\rho_X}, \infnorm{\rho_Y})^4 \log{\frac{1}{\delta}} \log\|b_k\|_1}.
\end{gather}
\end{proof}

%The quantum mutual information is defined as 

%\begin{align}
    %I(X,Y) &= S(\ro{XY}||\ro{X} \otimes \ro{y})\\
    %&= S(\ro{X}) + S(\ro{Y}) - S(\ro{XY})
%\end{align}

%where $S(\rho) = -Tr(\rho \,\, ln( \rho))$ is the von Neumann entropy and $S(\cdot||\cdot)$ quantum relative entropy, defined as $S(\rho||\sigma) := \Tr{\rho(\log{\rho} - \log{\sigma})}$. Unfortunately, due to the algebraic properties of the operator logarithm function, a straightforward calculation of the gradient of the objective function is impossible, as no closed form exists. In this paper we provide three approaches to working around this problem, with varying degrees of complexity and differing error properties. The first approach in the order of descending difficulty is approximating the matrix logarithm function with a polynomial/fourier expansion, the second approach is finding a variational bound on the objective functions in terms of R\'enyi-2 and R\'enyi-$\frac{1}{2}$ entropies; while the last approach is using the much more easier to calculate R\'enyi-2 entropy as a substitute for the standard quantum relative entropy.

\subsection{Approximations to the QIB Objective Function}
The previous results show a method that can be used to estimate the quantum information bottleneck using numerical differentiation of the objective function.  For the purposes of training, however, we can alternatively aim to minimize other objective functions that are known to upper bound the QIB function.  Specifically, by proving an upper bound on the QIB quantity we can then minimize that upper bound rather than the QIB objective function. The lower bounds can be used to estimate the quality of the approximation.

Below we present alternative methods that can be used to estimate this quantity that are based on R\'enyi entropies as well as the measured R\'enyi entropy.  The bounds for evaluation of these gradient approximations are worse than those for  the numerical approach considered above and as such we summarize the results below but leave the proofs in the appendix for the interested reader.  The most important point behind these results is that while these approaches can be used to bound the QIB function by a simpler function in certain cases, the increased analytical tractability of these approaches does not necessarily make the corresponding models easier to train than the numerical differentiation method introduced above.  These results show, somewhat surprisingly, that the methods that are best for analytically estimating the entropy need not yield analytic gradients that are easy to evaluate on a quantum computer.

\subsubsection{R\'enyi 2-entropy and divergence}
In general the QIB objective function is difficult to evaluate because of its dependence on matrix logarithms of the density matrix.  A standard approach for simplifying such problems in quantum information theory is to replace the von Neumann entropies considered with R\'enyi entropies, which are easier to compute and analyze ans they eschew the need for the density matrix.  Further, standard relations exist that allow us to upper bound quantities such as the relative entropy with the R\'enyi divergence.  Specifically, for the case of the R\'enyi 2-divergence, $D_2(\rho || \sigma)$,
\begin{equation}
    S(\rho|| \sigma)\le D_2(\rho || \sigma) \equiv -\ln({\rm Tr}(\rho^2 \sigma^{-1})) 
\end{equation}
As the matrix logarithm is replaced here by an ordinary logarithm, we no longer need to resort to series expansions to find an expression for the derivative.  Further, in the event that the distributions $\rho_X,\rho_{\tilde{X}}$ are thermal distributions of the form $e^{-H}/{\rm Tr}(e^{-H})$ for some Hermitian operator $H$ then there exists an analytic expression for the inverses as well in the objective function.  This can be especially valuable when trying to learn a thermal state using a quantum Boltzmann machine~\cite{kieferova2017tomography}.  Other choices of $\alpha$ and different flavors of R\'enyi entropy also have fractional powers of density matrices, and are not considered because they are more challenging to evaluate.

We show in the appendices standard inequalities on R\'enyi entropies and applications of Jensen's inequality lead to the following upper and lower bounds on the QIB objective function:
%\subsubsection{Bounds from R\'enyi-2 Entropy}

\begin{lemma}\label{lem:QIBapprox}
For all quantum channels $\Phi$, $\beta\ge 0$ and input states $\rho_{XYR}$ we have that the QIB objective function can be bounded above and below by
\begin{align*}
1)&\qquad \mathcal{L}(\Phi) \le - \ln \left({\rm Tr}(\rho_{X\widetilde{X}}^2 (\rho_{X}^{-1} \otimes \rho_{\widetilde{X}}^{-1}))\right) - \beta\frac{ {\rm Tr}((\rho_{\widetilde{X}Y}-\rho_{\widetilde{X}}\otimes \rho_{Y})^2)}{2}%:= \mathcal{L}_u(\Phi) 
\nonumber\\
2)&\qquad \mathcal{L}(\Phi) \ge \beta \ln({\rm Tr}(\rho_{\widetilde{X}Y}^2  (\rho_{\widetilde{X}}^{-1} \otimes \rho_{Y}^{-1}))) + \frac{ {\rm Tr}((\rho_{X\widetilde{X}}-\rho_{X}\otimes \rho_{\widetilde{X}})^2)}{2},  \nonumber
\end{align*}
\end{lemma}
Proof of this lemma is provided in Appendix \ref{sec:regularrenyi}.  These results provide a pair of useful bounds that allow us to strip away the matrix logarithms from the result, which makes the formulas much easier to think about and further allows us to express analytical forms for the gradients.  This raises the hope that these expressions may be useful also for optimizing the QIB objective on a quantum computer.  We show in Appendix \ref{sec:regularrenyi} that the query complexity for evaluating the gradient of the upper and lower bounds within error $\epsilon$ and failure probability $1-\delta$ is 
\[\bigotilde{d^2 \frac{1}{\lambda^6} \frac{1}{\epsilon} \left(d^3  \infnorm{\partial \rho} + \frac{1}{\lambda} \infnorm{\partial \sigma} \right)\log\frac{1}{\delta}}. \]  Here $d$ is the Hilbert-space dimension and $\sigma := \ro{X} \otimes \ro{\widetilde{X}}$

This bound shows that for low-dimensional spaces the gradient of the R\'enyi $2$-divergence can be computed efficiently; however, the best unconditional bounds that we can prove scale poorly with dimension.  Given stronger assumptions it may be possible to provide efficient training algorithms for this objective function in cases where $d$ is exponentially large, but the proof techniques that we employ are not capable of showing polynomial bounds in the general case.  This means that while the $2$-divergence is perhaps conceptually simpler, the training models using such approximations is not necessarily cheaper than differentiating a series expansion for the QIB objective function.

\subsubsection{Measured R\'enyi Entropy}

Another approach is to approximate the quantum relative entropy using another measure of relative entropy that is easier to calculate. However, the multiple definitions of quantum mutual information

\begin{align}
    I(A;B) &= S(\rho_A) + S(\rho_B) - S(\rho_{AB}) \label{rdef1}\\
    &= D(\rho_{AB}||\rho_A \otimes \rho_B)\label{rdef2} \\
    &= \operatorname{min}_{\sigma_B} D(\rho_{AB}||\rho_A \otimes \sigma_B)\label{rdef3}
\end{align}
are not equivalent anymore once we replace the entropy and divergence terms with their R\'enyi counterparts. Doing this swap in \eqref{rdef1} is desirable because all quantities to be calculated are polynomials of density matrices, but the resulting quantity fails to replicate desirable properties of $I(A;B)$ such as positivity, and using the definition in \eqref{rdef3} is undesirable due to the cost of the optimization problem added to each mutual information term. 

Reference \cite{measuredrenyi_cirac} proposes addressing these problems by using the  measured R\'enyi entropy, defined as 
\begin{equation}
    D_\alpha^\mathbb{M}(\rho || \sigma) := \operatorname{sup}_{\chi, M} D_\alpha(P_{\rho,M}||P_{\sigma,M})
\end{equation}
here $\chi$ is a finite set and $M$ is a POVM defined on $\chi$. Which can be interpreted as the classical R\'enyi divergence of the measurement probability distributions of the given density matrices, optimized over all possible measurements. \cite{measuredrenyi_tomamichel} proves that the optimum value can be attained via rank 1 projective measurements and provides an alternative variational form: 

\begin{equation}
D_\alpha^\mathbb{M}(\rho || \sigma) = \frac{1}{\alpha - 1} \operatorname{log} \operatorname{sup}_{\omega > 0} \left[\alpha \Tr{\rho \omega^{\alpha-1}} + (1-\alpha) \Tr{\alpha \omega^\alpha} \right]
\end{equation}
which has a closed from solution for $\alpha = 2$:
\begin{equation}
    D_2^\mathbb{M}(\rho || \sigma) = \operatorname{log} \left[\Tr{\rho \Phi^{-1}(\rho)} - \frac{1}{4} \Tr{\rho (\Phi^{-1}(\rho))^2} \right] \label{measfinal}
\end{equation}
where $\Phi(\omega) := \omega \sigma + \sigma \omega$. This quantity is strictly positive, and obeys the data processing inequality \cite{measuredrenyi_tomamichel}. The function $\Phi^{-1}(\rho)$ can be expressed in closed form considering the Lyapunov equation:

\begin{gather}
    \rho = \omega \sigma + \sigma \omega \\
    \omega \sigma + \sigma \omega^\dagger - \rho = 0\\
    (-\sigma) \omega + \omega (-\sigma)^H + \rho = 0 \\
    \omega = \Phi^{-1}(\rho) = \int_0^\infty e^{-s\sigma} \rho e^{-s \sigma} ds
\end{gather}

Unfortunately, na\"ive algorithms for approximating these integrals prove to be prohibitively costly.  Specifically, we find that a method exists for estimating them that runs in time $\bigo{\left(\frac{1}{\lambda} \right)^\frac{1}{\lambda}}$ where $\lambda$ is the minimum eigenvalue of $\sigma$. This shows that while the measured R\'enyi entropy is perhaps a conceptually simpler objective function to estimate, it is also not conceptually easier to evaluate and na\"ive bounds diverge quickly with $1/\lambda$ as $\lambda\rightarrow 0$.  This further reinforces our claim that analytic methods for estimating the QIB objective function do not necessarily lead to cheaper gradient evaluation than than the QIB function.  Refer to \ref{sec:measuredrenyi} for a detailed analysis of the approximation process.

\section{Conclusion}
We have provided in this paper methods for computing as well as optimizing the quantum information bottleneck function for a quantum channel.  We find that the query complexity, quantified by the number of accesses to a purification of the density operator provided to the quantum channel, scales polynomially provided that the non-zero eigenvalues of the density operator are not small.  We further find that analytic expressions for the gradients, such as bounds based on the R\'enyi divergence and also the measured R\'enyi divergence, can also be evaluated.  However, the algorithm discovered for finding the gradients of the measured R\'enyi does not scale inverse polynomially with the measured R\'enyi entropy and thus the question of whether efficient training algorithms exist for this objective function remains open.

This work shows that we can directly train a parameterized quantum channel (i.e. a quantum neural network) to optimize the amount of relevant information that passes through a bottleneck in a quantum process.  Further, this work also shows that the calculation of the training objective function can also be performed in polynomial time under reasonable assumptions about the spectrum of the operators.  This approach therefore provides a new and potentially powerful method for training quantum neural networks that not only differs from existing approaches, but also is strongly motivation from an information theoretic perspective.

There are many open questions that remain about the QIB.  The most obvious question is whether the optimization landscape for the QIB contains the same barren plateaus that other training objectives can have.  This is relevant because in the event that the gradients are small the cost of gradient evaluation will grow.  In addition, there is a question of what the practical benefits of training a quantum neural network according to QIB would be relative to existing loss functions.  This is particularly relevant as the parameter $\beta$ remains free  and thus can be chosen to optimize the empirical performance of the bound.  Finally, the we do not know whether the algorithms provided here for computing the gradient or value of the QIB are optimal, which means that there could be further polynomial advantages yet to be discovered by optimizing the algorithms provided here while at the same  time it also shows that the limitations that quantum computers face in computing the quantity are unknown. 
%In the paper [] the authors also consider the faults in estimation of the classical mutual information, due to the finite-sample effects the boundary curve in the information plane ceases to be monotonically decreasing and instead has a minimum. The learning performance of a final network can be compared against this point to be  
Probing these questions may not only provide us with more practical means to understand the capacity of quantum neural networks, but also provide a deeper understanding of the nature of information flow within implementations of quantum channels.

\acknowledgements{This work was funded by the US Department of Energy, Office of Science, National Quantum Information Science Research Centers, Co-Design
Center for Quantum Advantage under contract number DE-SC0012704 and  additional research related to this work was supported by a Google Research Award.}
\appendix

%\section{Formulation of the channel calculation of derivatives}

\section{Approximations to the QIB function and its derivatives}
The purpose of this appendix is to provide detailed proofs of the validity of our bounds on the QIB function as well as estimates of the derivatives of these bounds.  We examine two classes of approximations here: first we consider upper and lower bounds provided by R\'enyi entropies and second we consider the measured R\'enyi divergence, which is an approximate approach that aims to deal with the ambiguities between the various statements of the mutual information when the divergences involved are replaced with R\'enyi divergences.
\subsection{R\'enyi Bounds}
\label{sec:regularrenyi}

The upper bound is primarily useful for our variational optimization whereas the lower bound provides an estimate of the maximum error that we can have in our approximate objective function.

\begin{proof}[Proof of Lemma~\ref{lem:QIBapprox}]
 The quantum information bottleneck quantity can be written in terms of R\'enyi divergences as 
\begin{equation}
    \mathcal{L}(\Phi) =  D_{1}(\rho_{X\widetilde{X}} \| \rho_{X}\otimes \rho_{\widetilde{X}}) - \beta D_1(\rho_{\widetilde{X}Y}\| \rho_{\widetilde X} \otimes \rho_Y).
\end{equation}
Next using the fact that the Renyi divergences are monotonically non-decreasing with respect to $\alpha$ and using Theorem 1.15 of~\cite{ohya2004quantum}, which states that $D_1(\rho\| \sigma) \ge \frac{1}{2} \|\rho-\sigma\|_1^2$, we have that
\begin{equation}
    D_2(\rho\|\sigma) \ge D_1(\rho\|\sigma) \ge \frac{1}{2} \|\rho-\sigma\|_1^2.
\end{equation}
This observation further implies that
\begin{align}
    \mathcal{L} &\le   D_{2}(\rho_{X\widetilde{X}} \| \rho_{X}\otimes \rho_{\widetilde{X}}) -  \frac{\beta}{2} \|\rho_{\widetilde{X}Y}-\rho_{\widetilde{X}}\otimes \rho_{Y}\|_1^2 \nonumber\\
    &\le D_{2}(\rho_{X\widetilde{X}} \| \rho_{X}\otimes \rho_{\widetilde{X}}) -  \frac{\beta}{2} \|\rho_{\widetilde{X}Y}-\rho_{\widetilde{X}}\otimes \rho_{Y}\|_2^2.\nonumber\\
    &\le D_{2}(\rho_{X\widetilde{X}} \| \rho_{X}\otimes \rho_{\widetilde{X}}) -  \frac{\beta}{2} {\rm Tr}((\rho_{\widetilde{X}Y}-\rho_{\widetilde{X}}\otimes \rho_{Y})^2).
\end{align}
Next we use the fact that the R\'enyi 2-divergence is
\begin{gather}
    D_2(\rho || \sigma) \equiv -\ln({\rm Tr}(\rho^2 \sigma^{-1}))
\end{gather}
to see that 
\begin{equation}
    \mathcal{L}(\Phi) \le -\ln({\rm Tr}(\rho_{X\tilde{X}}^2 \rho_X^{-1} \otimes \rho_{\tilde{X}}^{-1})) -  \frac{\beta}{2} {\rm Tr}((\rho_{\widetilde{X}Y}-\rho_{\widetilde{X}}\otimes \rho_{Y})^2):= \mathcal{L}_u(\Phi)
\end{equation}

The lower bound follows by repeating these steps, with the role of the first and second divergences in $\mathcal{L}(\Phi)$ switched.  Specifically,
\begin{align}
    \mathcal{L} &\ge   \frac{1}{2} \left\|\rho_{X\widetilde{X}} - \rho_{X}\otimes \rho_{\widetilde{X}}\right\|_1^2 - \beta D_1(\rho_{\widetilde{X}Y}\| \rho_{\widetilde X} \otimes \rho_Y) \nonumber\\
    &\ge \beta \ln({\rm Tr}(\rho_{\widetilde{X}Y}^2  (\rho_{\widetilde{X}}^{-1} \otimes \rho_{Y}^{-1}))) + \frac{ {\rm Tr}((\rho_{X\widetilde{X}}-\rho_{X}\otimes \rho_{\widetilde{X}})^2)}{2} := \mathcal{L}_l(\Phi)
    \end{align} 
\end{proof}

The R\'enyi-2 divergence is used here in preference to the full relative entropy largely for simplicity.  The fact that the trace is taken inside the logarithm makes differentiation of the quantity much easier and second it is known to be an upper bound on the relative entropy.
It also obeys the data processing inequality, which is vital for our method. Thus the R\'enyi-2 divergence is a well motivated and also a (comparably) easy to compute training loss function.  This approach has also been used extensively in other work involving training unitary quantum neural networks as well as Boltzmann machines~\cite{kieferova2017tomography,kieferova2021quantum}.

\begin{lemma}
Let $\sigma := \ro{X} \otimes \ro{\widetilde{X}}$.  The derivative of the upper bound $\mathcal{L}_u(\Phi)$ with respect to the parameter $\alpha_k$ is
\begin{align*}
\partial_{\alpha_k}{\mathcal{L}_u(\Phi)} &= -\frac{\Tr( \{\taurx,(\partial_{\alpha_k} \taurx)\} \sigi - \taurx^2 \sigi \sigma' \sigi)}{\Tr(\taurx^2 \sigma^{-1})} + \Tr{\frac{\partial \rho_{\widetilde{X}Y}}{\partial {\alpha_k}} \rho_{\widetilde{X}Y} + \rho_{\widetilde{X}Y} \frac{\partial \rho_{\tilde{X}Y}}{\partial {\alpha_k}}} \nonumber\\
    &+\Tr{\frac{\partial (\rho_{\tilde{X}}\otimes \rho_{Y})}{\partial {\alpha_k}} (\rho_{\tilde{X}}\otimes \rho_{Y}) + (\rho_{\tilde{X}}\otimes \rho_{Y}) \frac{\partial (\rho_{\tilde{X}}\otimes \rho_{Y})}{\partial {\alpha_k}}}\nonumber \\
    &- 2\Tr{\frac{\partial \rho_{\tilde{X}Y}}{\partial {\alpha_k}} (\rho_{\tilde{X}}\otimes \rho_{Y}) + \rho_{\tilde{X}Y} \frac{\partial (\rho_{\tilde{X}}\otimes \rho_{Y})}{\partial {\alpha_k}}}.
\end{align*}
\end{lemma}
\begin{proof}
We will need to work through a few calculations before proving this theorem, starting with calculating the gradient of the first term. Let us substitute $\tau := \rtp$ and $\sigma := \ro{X} \otimes \ro{\widetilde{X}}$ for brevity. 

\begin{align}
    \pai{D_2} &= -\frac{\pai{(\Tr(\taurx^2 \sigma^{-1}))}}{{\rm Tr}(\taurx^2 \sigma^{-1})}\nn
    &= -\frac{\Tr(\pai{(\taurx^2 \sigma^{-1})})}{\Tr(\taurx^2 \sigma^{-1})} \nn
    &= -\frac{\Tr((\pai{\taurx^2}) \sigma^{-1} + \taurx^2 \pai{\sigma^{-1}} )} {\Tr(\taurx^2 \sigma^{-1})} \nn
    &= -\frac{\Tr( (\partial_k \taurx) \taurx \sigi + \tau (\partial_k \taurx) \sigi - \taurx^2 \sigi \sigma' \sigi)} {\Tr(\taurx^2 \sigma^{-1})} \label{eq:ider} \nn
    &= -\frac{\Tr( \{\taurx,(\partial_k \taurx)\} \sigi - \taurx^2 \sigi \sigma' \sigi)} {\Tr(\taurx^2 \sigma^{-1})} 
\end{align}

 For the second term:

\begin{align}
    \frac{\partial}{\partial {\alpha_k}} \Tr{(\rho_{\widetilde{X}Y}-\rho_{\widetilde{X}}\otimes \rho_{Y})^2}&= \Tr{\frac{\partial \rho_{\widetilde{X}Y}}{\partial {\alpha_k}} \rho_{\widetilde{X}Y} + \rho_{\widetilde{X}Y} \frac{\partial \rho_{\widetilde{X}Y}}{\partial {\alpha_k}}} \nonumber\\
    &+\Tr{\frac{\partial (\rho_{\widetilde{X}}\otimes \rho_{Y})}{\partial {\alpha_k}} (\rho_{\widetilde{X}}\otimes \rho_{Y}) + (\rho_{\widetilde{X}}\otimes \rho_{Y}) \frac{\partial (\rho_{\widetilde{X}}\otimes \rho_{Y})}{\partial {\alpha_k}}}\nonumber \\
    &- 2\Tr{\frac{\partial \rho_{\widetilde{X}Y}}{\partial {\alpha_k}} (\rho_{\widetilde{X}}\otimes \rho_{Y}) + \rho_{\widetilde{X}Y} \frac{\partial (\rho_{\widetilde{X}}\otimes \rho_{Y})}{\partial {\alpha_k}}}
    \label{eq:method2secondterm}
\end{align}

The total derivative is

\begin{align}
    \partial_{\alpha_k}{\mathcal{L}_u(\Phi)} &= -\frac{\Tr( \{\taurx,(\partial_k \taurx)\} \sigi - \taurx^2 \sigi \sigma' \sigi)}{\Tr(\taurx^2 \sigma^{-1})} + \Tr{\frac{\partial \rho_{\widetilde{X}Y}}{\partial {\alpha_k}} \rho_{\widetilde{X}Y} + \rho_{\widetilde{X}Y} \frac{\partial \rho_{\tilde{X}Y}}{\partial {\alpha_k}}} \nonumber\\
    &+\Tr{\frac{\partial (\rho_{\tilde{X}}\otimes \rho_{Y})}{\partial {\alpha_k}} (\rho_{\tilde{X}}\otimes \rho_{Y}) + (\rho_{\tilde{X}}\otimes \rho_{Y}) \frac{\partial (\rho_{\tilde{X}}\otimes \rho_{Y})}{\partial {\alpha_k}}}\nonumber \\
    &- 2\Tr{\frac{\partial \rho_{\tilde{X}Y}}{\partial {\alpha_k}} (\rho_{\tilde{X}}\otimes \rho_{Y}) + \rho_{\tilde{X}Y} \frac{\partial (\rho_{\tilde{X}}\otimes \rho_{Y})}{\partial {\alpha_k}}}
    \label{eq:method2derivative}
\end{align}
\end{proof}

Next in order to evaluate the derivatives we need to introduce approximations to the matrix inverse or pseudoinverse.  
In order to calculate each of the trace terms in \ref{eq:ider}, we first need to calculate $\sigps$. We use the method from~\cite{Childs_2017}. This method aims to express the inverse of a matrix as a sum of Chebyshev polynomials using an approximation of the function $f(x) = 1/x$, given that eigenvalues are in the domain $D_\kappa = [-1, \frac{1}{\kappa}) \cup (\frac{1}{\kappa},1]$.  Since this method maps 0 eigenvalues to 0; the algorithm still produces a result that's valid in the support of $\sigma$.

The following provides an approximation to the inverse.

\begin{lemma}[Lemma 14 of \cite{Childs_2017}]
\label{lem:childs}
Let $x \in D_\kappa$ then there exists a polynomial $f(x)$
\begin{gather}
    f(x) = 4 \sum_{j=0}^{j_0} (-1)^j \left[\frac{\sum_{i=j+1}^{b} {2b \choose b+i}}{2^{2b}} \right] \mathcal{T}_{2j+1}(x)
\end{gather}
that is $2\epsilon$-close to $1/x$ in $D_\kappa$, meaning that $|\frac{1}{x}-f(x)| \leq 2\epsilon$ for all $x$ in $[1/\kappa,1]$ where  $b=\kappa^2 \log(\frac{\kappa}{\epsilon})$ and $j_0 = \sqrt{b\, \log(\frac{4b}{\epsilon})}$. $\mathcal{T}_k$ are Chebyshev polynomials of first kind. They are defined by the recurrence relations $\mathcal{T}_0 = 1$, $\mathcal{T}_1 = x$ and $\mathcal{T}_n(x) = 2x\mathcal{T}_{n-1}(x) - \mathcal{T}_{n-2}(x)$.
\end{lemma}
This then leads to the following theorem.
\begin{theorem}
\label{thm:method2main}
The derivative of the variational upper bound $\mathcal{L}_u(\Phi)$ with respect to $\alpha_k$ can be calculated within error $\epsilon$ with $\bigotilde{d^2 \frac{1}{\lambda^6} \frac{1}{\epsilon} \left(d^3  \infnorm{\partial \rho} + \frac{1}{\lambda} \infnorm{\partial \sigma} \right)\log\frac{1}{\delta}}$ queries to an oracle that provides the purifications of the relevant density matrices  and with probability $p \geq 1 - \delta$, $\delta > 0$.
\end{theorem}
\begin{proof}
We can use this formulation of the derivative of the channel to express the derivative of the variational bound as a sum of polynomial terms, which then can be calculated efficiently using the Hadamard test. The two sources of error are the deterministic error from approximating $\sigi$ and probabilistic error from estimating the traces. Let us define $N := \Tr( \{\taurx,(\partial_k \taurx)\} \sigi - \taurx^2 \sigi \sigma' \sigi)$ and $D := \Tr(\taurx^2 \sigma^{-1})$. The bounds on the deterministic errors for $N$ and $D$ are given by:

\begin{align}
    \epsilon_N &\leq 2 \varepsilon \infnorm{\partial_k \taurx} + \frac{2 \varepsilon \infnorm{\partial_{\alpha_k}\sigma}}{\lambda},\\
    \epsilon_D &\leq \varepsilon.
\end{align}
Here we defined $\varepsilon$ as the error parameter in Lemma~\ref{lem:childs} and $\lambda$ is the minimum eigenvalue of $\sigma$. Substituting in the power series for the inverse into the first term in \refeq{eq:method2derivative} expansion then gives us the approximation

\begin{gather}\label{eq:seconsubstituted}
    \frac{N}{D} \approx \frac{\widetilde{N}}{\widetilde{D}} = \frac{\Tr( \{\taurx,(\partial_{\alpha_k} \taurx)\} (\sum_{i=0}^{j_0}a_i \sigma^i)  - \taurx^2 (\sum_{i=0}^{j_0}a_i \sigma^i) \sigma' (\sum_{i=0}^{j_0}a_i \sigma^i))}{\Tr(\taurx^2\sum_{i=0}^{j_0}a_i \sigma^i)}
\end{gather}

Expressing $\partial_{\alpha_k} \rho_{R\widetilde{X}}= \sum_j^{N_u} b_{jk} V_j$ and $\rho_R \otimes \partial_{\alpha_k} \rho_{\widetilde{X}}= \sum_j^{N_u} c_{jk} V_j$ in a unitary basis, the expression in turn reads:

\begin{align}
    \frac{\widetilde{N}}{\widetilde{D}} &= \frac{\sum_{i=0}^{j_0}\sum_j^{N_u} a_i b_{jk}\Tr( \taurx V_j \sigma^i) + \sum_{i=0}^{j_0}\sum_j^{N_u} a_i b_{jk}\Tr(V_j \taurx \sigma^i)  - \sum_{i=0}^{j_0} \sum_{n=0}^{j_0}\sum_j^{N_u} c_{jk} a_i  a_n\Tr(\taurx^2 \sigma^i V_{jk}  \sigma^n))}{\sum_{i=0}^{j_0}a_i \Tr(\taurx^2 \sigma^i)}
    \label{eq:a14}
\end{align}

With the ability to evaluate traces of the form $\Tr{\prod_i\rho_i U_i}$ with error $\epsilon_T$, the probabilistic error margins for $\widetilde{N}$ and $\widetilde{D}$ read using the bounds on the terms in the series approximation to the inverse in Lemma~\ref{lem:childs}.

\begin{align}
    \varepsilon_N &\leq N_u \epsilon_T (2 \{b_{jk}\}_{max}\sum_i |a_i| + \{c_{jk}\}_{max}\sum_i |a_i|^2) \leq 4 N_u \epsilon_T \left( \{b_{jk}\}_{max}  \kappa \sqrt{\log \frac{\kappa}{\varepsilon}} + \{c_{jk}\}_{max} \kappa^2 \log \frac{\kappa}{\varepsilon} \right)\\  
    \varepsilon_D &\leq \epsilon_T \sum_i |a_i| = 2 \epsilon_T \kappa \sqrt{\log \frac{\kappa}{\varepsilon}}
\end{align}

With probability greater than 
\begin{equation}\label{eq:succprob}
    P_{succ}\ge 1 - e^{-\frac{n_{AE}}{24}} (N_u(j_0+j_0^2) + j_0).
\end{equation} 

The bound is derived by considering that the probability of failure for a single trace evaluation is upper bounded by $p_{fail} \leq e^{-\frac{n_{AE}}{24}}$ from \eqref{eq:failurebound}, where $n_{AE}$ is a parameter of choice, and using the fact that $N_u(j_0^2 + j_0)+j_0$ traces are evaluated. The equation follows from the union bound $P(\bigcup_i A_i) \leq \sum_i P(A_i)$. We then evaluate the cumulative bound on the error:

\begin{align}
    \epsilon_{total}&:= \abs{\frac{N}{D} - \frac{N + \epsilon_N + \varepsilon_N}{D+ \epsilon_D + \varepsilon_D}}\\
    &\leq \frac{2 \abs{\epsilon_N + \varepsilon_N}}{D} + \frac{\abs{N}\abs{\epsilon_D + \varepsilon_D}}{D^2}\\
    &= 2 \infnorm{\sigma} \left(2 \varepsilon \infnorm{\partial_{\alpha_k}\taurx} + \frac{\varepsilon \infnorm{\partial_{\alpha_k}\sigma}}{\lambda} + 4\epsilon_T N_u \left(\{b_{jk}\}_{max}  \kappa \sqrt{\log \frac{\kappa}{\varepsilon}} + \{c_{jk}\}_{max} \kappa^2 \log \frac{\kappa}{\varepsilon} \right)\right) \nonumber\\
    &\qquad+ \infnorm{\sigma} (\frac{2\infnorm{\partial_{\alpha_k}\taurx}}{\lambda} + \frac{\infnorm{\partial_{\alpha_k}\sigma}}{\lambda^2}) (\varepsilon + 2\epsilon_T \kappa \sqrt{\log \frac{\kappa}{\varepsilon}}) \\
    &= \varepsilon (\infnorm{\partial_{\alpha_k}\taurx} (4+2\kappa) + \infnorm{\partial_{\alpha_k}\sigma} (2\kappa + \kappa^2)) \nn
    &\hspace{20pt}+\epsilon_T (8 N_u \left(\{b_{jk}\}_{max}  \kappa \sqrt{\log \frac{\kappa}{\varepsilon}} + \{c_{jk}\}_{max} \kappa^2 \log \frac{\kappa}{\varepsilon} \right) + 4\infnorm{\partial_{\alpha_k}\taurx}\kappa^2\sqrt{\log \frac{\kappa}{\varepsilon}} + 2\infnorm{\partial_{\alpha_k}\sigma} \kappa^3 \sqrt{\log \frac{\kappa}{\varepsilon}})
\end{align}

Where we substituted in the definitions of $\epsilon_N$, $\varepsilon_N$, $\epsilon_D$ and $\varepsilon_D$ and in the last line rearranged the terms into groups of $\varepsilon$ and $\epsilon_T$. Setting both the $\epsilon_T$ and $\varepsilon$ terms to be equal to $\epsilon/2$  sets the global error to $\epsilon$, giving us the parameters $\varepsilon$ and $\epsilon_T$ in terms of the global error:

\begin{align}
    \varepsilon &\leq \frac{\epsilon}{2 (\infnorm{\partial_{\alpha_k}\taurx} (4+2\kappa) + \infnorm{\partial_{\alpha_k}\sigma} (2\kappa + \kappa^2))}\\
    \epsilon_T &\leq \frac{\epsilon}{2(8N_u \left(\{b_{jk}\}_{max}  \kappa \sqrt{\log \frac{\kappa}{\varepsilon}} + \{c_{jk}\}_{max} \kappa^2 \log \frac{\kappa}{\varepsilon} \right) + 4\infnorm{\partial_{\alpha_k}\taurx}\kappa^2\sqrt{\log \frac{\kappa}{\varepsilon}} + 2\infnorm{\partial_{\alpha_k}\sigma} \kappa^3 \sqrt{\log \frac{\kappa}{\varepsilon}})}
\end{align}

We can see that $\bigo{N_u(j_0^2 + j_0^4) + j_0^2}  \in \bigo{N_u j_0^4}$ oracle accesses to the $\sigma$ oracle is necessary to evaluate all of the traces in the expression \eqref{eq:a14}, given that the method we're using to evaluate traces requires $n_i$ copies of $\rho_i$ to evaluate $\Tr(\prod \rho_i^{n_i} U_i)$. We had shown in \eqref{eq:succprob} that the failure probability was upper bounded such that

\begin{align}
    \delta &\leq e^{-\frac{n_{AE}}{24}} (N_u(j_0+j_0^2) + j_0)
\end{align}

Which gives us the lower bound on $n_{AE}$:

\begin{align}
    n_{AE} \geq 24 \log(\frac{N_u(j_0+j_0^2) + j_0}{\delta})
\end{align}

The total query complexity comes from multiplying the base algorithm cost with the number of repetitions of amplitude estimation, in the procedure described in Appendix \ref{sec:swaptest}. 

\begin{align}
    \textrm{QC} &= \bigo{N_u j_0^4}\cdot\bigo{\epsilon_T^{-1}}\cdot\bigo{n_{AE}}\\
    &= \bigo{N_u \frac{1}{\lambda^6}\frac{1}{\epsilon}\log\frac{1}{\delta}\log(\frac{1}{\epsilon \lambda^2})^4 \left(N_u b_{max} \log\frac{1}{\epsilon \lambda} + (\infnorm{\partial \rho} + \frac{1}{\lambda} \infnorm{\partial \sigma})\sqrt{\log\frac{1}{\epsilon \lambda}} \right)}\\
    &= \bigotilde{N_u \frac{1}{\lambda^5} \frac{1}{\epsilon}\log\frac{1}{\delta} \left(N_u (b_{max}\frac{1}{\lambda}+c_{max}) + \infnorm{\partial \rho} + \frac{1}{\lambda} \infnorm{\partial \sigma} \right)}
\end{align}
Here $b_{max}$ is the maximum of the absolute coefficients of the expansion $\partial_{\alpha_k} \taurx = \sum_{i=1}^{N_u} b_{ik} V_i$. Assuming an orthogonal basis of $d^2$ elements, $N_u = d^2$, $b_{ik} = \langle \partial_{\alpha_k} \taurx, V_i \rangle = \Tr(V_i^\dagger (\partial_{\alpha_k} \taurx))$.

\begin{align}
    b_{ik} &= \Tr(V_i^\dagger (\partial_{\alpha_k} \taurx)),\\
    &\leq \infnorm{\partial_{\alpha_k} \taurx} \Tr(\abs{V_i^\dagger}),\\
    &\leq \infnorm{\partial_{\alpha_k} \taurx} d.
\end{align}

Similarly:

\begin{align}
    c_{ik} &= \Tr(V_i^\dagger (\partial_{\alpha_k} \sigma)),\\
    &\leq \infnorm{\partial_{\alpha_k} \sigma} \Tr(\abs{V_i^\dagger}),\\
    &\leq \infnorm{\partial_{\alpha_k} \sigma} d.
\end{align}

Since the bound is independent of the i index, it also bounds the maximum value. Thus the upper bound on the QC is:

\begin{align}
    \bigotilde{d^2 \frac{1}{\lambda^5} \frac{1}{\epsilon} \left(d^3  \infnorm{\partial \taurx}\frac{1}{\lambda} + (\frac{1}{\lambda}+d^3) \infnorm{\partial \sigma} \right)\log\frac{1}{\delta}}.
\end{align}
\end{proof}

\subsection{Measured Renyi Divergence}
\label{sec:measuredrenyi}
Another approach is to approximate the quantum relative entropy using another measure of relative entropy that is easier to calculate. However, the multiple definitions of quantum mutual information

\begin{align}
    I(A;B) &= S(\rho_A) + S(\rho_B) - S(\rho_{AB}) \label{rdef1}\\
    &= D(\rho_{AB}||\rho_A \otimes \rho_B)\label{rdef2} \\
    &= \operatorname{min}_{\sigma_B} D(\rho_{AB}||\rho_A \otimes \sigma_B)\label{rdef3}
\end{align}
are not equivalent anymore once we replace the entropy and divergence terms with their R\'enyi counterparts. Doing this swap in \eqref{rdef1} is desirable because all quantities to be calculated are polynomials of density matrices, but the resulting quantity fails to replicate desirable properties of $I(A;B)$ such as positivity, and using the definition in \eqref{rdef3} is undesirable due to the cost of the optimization problem added to each mutual information term. 

The work in \cite{berta2017variational} proposes addressing these problems by using the  measured R\'enyi entropy, defined as 
\begin{equation}
    D_\alpha^\mathbb{M}(\rho || \sigma) := \operatorname{sup}_{\chi, M} D_\alpha(P_{\rho,M}||P_{\sigma,M})
\end{equation}
here $\chi$ is a finite set and $M$ is a POVM defined on $\chi$. Which can be interpreted as the classical R\'enyi divergence of the measurement probability distributions of the given density matrices, optimized over all possible measurements. 
Reference \cite{berta2017variational} proves that the optimum value can be attained via rank 1 projective measurements and provides an alternative variational form: 

\begin{equation}
D_\alpha^\mathbb{M}(\rho || \sigma) = \frac{1}{\alpha - 1} \operatorname{log} \operatorname{sup}_{\omega > 0} \left[\alpha \Tr{\rho \omega^{\alpha-1}} + (1-\alpha) \Tr{\alpha \omega^\alpha} \right]
\end{equation}
which has a closed from solution for $\alpha = 2$:

\begin{equation}
    D_2^\mathbb{M}(\rho || \sigma) = \operatorname{log} \left[\Tr{\rho \Phi^{-1}(\rho)} - \frac{1}{4} \Tr{\rho (\Phi^{-1}(\rho))^2} \right] \label{measfinal}
\end{equation}

where $\Phi(\omega) := \omega \sigma + \sigma \omega$. This quantity is strictly positive, and obeys the data processing inequality. The function $\Phi^{-1}(\rho)$ can be expressed in closed form considering the Lyapunov equation:

\begin{gather}
    \rho = \omega \sigma + \sigma \omega \\
    \omega \sigma + \sigma \omega^H - \rho = 0\\
    (-\sigma) \omega + \omega (-\sigma)^H + \rho = 0 \\
    \omega = \Phi^{-1}(\rho) = \int_0^\infty e^{-s\sigma} \rho e^{-s \sigma} ds
\end{gather}
%\NW{what is the superscript H here?}{\color{red}B:Hermitian conjugate}

Thus \eqref{measfinal} takes the form:

\begin{align}
    D_2^\mathbb{M}(\rho || \sigma) &= \operatorname{log} \left[-\Tr{\rho \int_0^\infty e^{-s\sigma} \rho e^{-s \sigma} ds} - \frac{1}{4} \Tr{\rho (- \int_0^\infty e^{-s\sigma} \rho e^{-s \sigma} ds)^2} \right] \\
    &= \operatorname{log} \left[-\int_0^\infty ds \Tr{\rho  e^{-s\sigma} \rho e^{-s \sigma}} - \frac{1}{4} \Tr{\rho ( \int_0^\infty e^{-s\sigma} \rho e^{-s \sigma} ds)^2} \right] \\
    &= \operatorname{log} \left[\int_0^\infty ds \Tr{\rho  e^{-s\sigma} \rho e^{-s \sigma}} - \frac{1}{4} \int_0^\infty \int_0^\infty ds ds' \Tr{\rho e^{-s\sigma} \rho e^{-(s+s')\sigma} \rho e^{-s' \sigma} } \right]
\end{align}

Since logarithm is a monotonic function, optimizing the logarithm is the same as optimizing the argument. Let $D_2^\mathbb{M}(\rho || \sigma) = \operatorname{log} Q_2^\mathbb{M}(\rho || \sigma)$. Utilizing the cyclic property of the trace, the derivative of $Q_2^\mathbb{M}(\rho || \sigma)$ can be expressed as 

\begin{align}
    \frac{\partial Q_2^\mathbb{M}}{\partial \alpha} &=2 \int_0^\infty ds \Tr{\frac{\partial \rho}{\partial \alpha}  e^{-s\sigma} \rho e^{-s \sigma}}- 2 \int_0^\infty ds s \int_0^1 dt \Tr{e^{-st\sigma}A e^{-s(1-t)\sigma} \rho e^{-s\sigma} \rho} \nn
    & - \frac{1}{4}  \int_0^\infty \int_0^\infty ds ds' \frac{\partial}{\partial \alpha}\left[\Tr{\rho e^{-s\sigma} \rho e^{-(s+s')\sigma} \rho e^{-s' \sigma} }\right] \label{qder}
\end{align}

Let us tackle the first two terms. Define $I = \int_0^\infty ds s \int_0^1 dt \Tr{e^{-st\sigma}A e^{-s(1-t)\sigma} \rho e^{-s\sigma} \rho}$. First we're going to approximate $e^{-A}$ with a power series, then we will approximate the integral with a Monte Carlo sampling method and order statistics, and then evaluate the approximation using the Extended swap test we introduced in the previous section.

Our goal is to calculate integrals of the form

\begin{align}
    I &= \int_0^\infty \Tr{e^{-st\sigma} A e^{-s(1-t)\sigma \rho e^{-s\sigma\rho}}} ds
\end{align}

This is done in a few steps, first we introduce a cutoff $L$ to the integral. The error caused by this approximation is calculated

\begin{align}
    \abs{I - \int_0^L \Tr{e^{-st\sigma} A e^{-s(1-t)\sigma \rho e^{-s\sigma\rho}}}ds} &= \abs{\int_L^\infty \Tr{e^{-st\sigma} A e^{-s(1-t)\sigma} \rho e^{-s\sigma\rho}}ds} \\
    &\leq \int_L^\infty \abs{\Tr{e^{-st\sigma} A e^{-s(1-t)\sigma} \rho e^{-s\sigma}\rho}}ds\\
    &\leq \int_L^\infty \infnorm{\partial_{\alpha_k}\rho} e^{-2s\lambda} \infnorm{\rho} ds\\
    &= \infnorm{\partial_{\alpha_k}\rho} \infnorm{\rho} \frac{e^{-2L\lambda}}{2\lambda}
\end{align}

Where we have defined $\lambda$ as the smallest eigenvalue of $\sigma$. We bound the error with $\epsilon_1$, which gives us a lower bound on L:

\begin{align}
    L \geq \frac{1}{2\lambda} \log \frac{\infnorm{\partial_{\alpha_k}\rho} \infnorm{\rho}}{2 \lambda \epsilon_1}
\end{align}

In the next step, we approximate the integral with a divided sum using a Newton-Cotes method. Dividing the interval $[0,L]$ into $n$ slices, we have \cite{abramowitz1964handbook} :

\begin{align}
    \int_0^L f(s) ds = \sum_{i=0}^{n} w_i f(\frac{L}{n}i) - \frac{1}{12}\frac{L^3}{n^2}f^{''}(\xi)
\end{align}

Here, $\xi$ is the value in the interval $[0,L]$ that maximizes the derivative. For this particular estimator $w_0 = w_n = \frac{1}{2}$ and $w_i = 1$ otherwise. Thus as an estimate of the error, we have

\begin{align}
  \abs{\int_0^L f(s) ds - \sum_{i=0}^{n} w_i f(\frac{L}{n}i)} &= \frac{1}{12}\frac{L^3}{n^2}\abs{f^{''}(\xi)}   
\end{align}

The derivative of $f(s) := \Tr{e^{-st\sigma} A e^{-s(1-t)\sigma \rho e^{-s\sigma\rho}}}$ obeys the inequality

\begin{align}
    \abs{f^{(n)}(s)} &\leq (2 \infnorm{\sigma})^n \infnorm{\partial_{\alpha_k}\rho} e^{-2s\lambda}\\
    &\leq (2 \infnorm{\sigma})^n \infnorm{\partial_{\alpha_k}\rho}
\end{align}

Thus

\begin{align}
    \abs{\int_0^L f(s) ds - \sum_{i=0}^{n} w_i f(\frac{L}{n}i)} &\leq \frac{\infnorm{\sigma}^2}{3}\frac{L^3}{n^2}
\end{align}

Upper bounding this portion of the error with $\epsilon_2$, we obtain a lower bound on $n$:

\begin{align}
    n \geq \kappa_\sigma \sqrt{\frac{1}{24 \lambda\epsilon_2} \left(\log \frac{\infnorm{\partial_{\alpha_k}\rho} \infnorm{\rho}}{2 \lambda \epsilon_1} \right)^3}
\end{align}

Here we have expressed the conditioning number of $\sigma$ as $\kappa_\sigma$. The last deterministic approximation we'll make is approximating the exponentials of $\sigma$ with power series. 

\begin{align}
    e^{-\alpha \sigma} &= \sum_{n=0}^{\infty} \frac{(-\alpha \sigma)^n}{n!}\\
    &= \left(\sum_{n=0}^{K} + \sum_{n=K+1}^{\infty} \right) \frac{(-\alpha \sigma)^n}{n!}
\end{align}

The error in norm is then

\begin{align}
    \infnorm{e^{-\alpha \sigma} - \sum_{n=0}^{K}\frac{(-\alpha \sigma)^n}{n!}} &= \infnorm{\sum_{n=K+1}^{\infty}\frac{(-\alpha \sigma)^n}{n!}}\\
    &\leq \sum_{n=K+1}^{\infty}\frac{(\alpha \infnorm{\sigma})^n}{n!}\\
    &\leq e^{\alpha \infnorm{\sigma}} \frac{(\alpha \infnorm{\sigma})^{K+1}}{(K+1)!}
\end{align}

Set this error to be less than $\frac{1}{\epsilon_0}$.

\begin{align}
    e^{\alpha \infnorm{\sigma}} \frac{(\alpha \infnorm{\sigma})^{K+1}}{(K+1)!} &\leq \epsilon_0\\
    \alpha \infnorm{\sigma} + (K+1)\ln(\alpha \infnorm{\sigma}) - \ln((K+1)!) &\leq \ln \epsilon_0 \\
    \ln((K+1)!) - (K+1)\ln(\alpha \infnorm{\sigma}) - \alpha \infnorm{\sigma}&\geq \ln\frac{1}{\epsilon_0}\\
    1+(K+\frac{3}{2})\ln(K+1) - (K+1) - (K+1)\ln(\alpha \infnorm{\sigma}) &\geq \ln\frac{1}{\epsilon_0}+\alpha \infnorm{\sigma}\\
    (K+1)(\ln(K+1) -\ln(\alpha \infnorm{\sigma})) &\geq \ln\frac{1}{\epsilon_0}+\alpha \infnorm{\sigma}-1\\
    \frac{K+1}{\alpha \infnorm{\sigma}}(\ln(K+1) -\ln(\alpha \infnorm{\sigma})) &\geq \frac{1}{\alpha \infnorm{\sigma}}
    (\ln\frac{1}{\epsilon_0}+\alpha \infnorm{\sigma}-1)
\end{align}

Recognizing that the left side is in the form $ye^y$ where $y=\ln(K+1) -\ln(\alpha \infnorm{\sigma})$, we can use the Lambert-W functions to express the solution:

\begin{align}
    \ln(K+1) -\ln(\alpha \infnorm{\sigma}) &\geq W_0(\frac{1}{\alpha \infnorm{\sigma}}
    (\ln\frac{1}{\epsilon_0}+\alpha \infnorm{\sigma}-1))\\
    K+1 &\geq \alpha \infnorm{\sigma}\exp{W_0(\frac{1}{\alpha \infnorm{\sigma}}
    (\ln\frac{1}{\epsilon_0}+\alpha \infnorm{\sigma}-1))}
\end{align}

Using monotonicity of the $W_0$ function, we can simplify the expression at the cost of making it a bit less tight.

\begin{gather}
    K(\alpha) \geq \alpha \infnorm{\sigma}\exp{W_0(\frac{1}{\alpha \infnorm{\sigma}}
    \ln\frac{1}{\epsilon_0}+1)}
\end{gather}

For sufficiently large values of $x$, $\exp{W_0(1+x)}$ is smaller than $x$.

\begin{gather}
    K \geq \ln\frac{1}{\epsilon_0}\\
    K \in \bigo{\ln\frac{1}{\epsilon_0}}
\end{gather}

Now let us bound the error we get in the trace due to using these approximations:

\begin{align}
    &\abs{\Tr{e^{-st\sigma} A e^{-s(1-t)\sigma} \rho e^{-s\sigma}\rho} - \Tr{\sum_{i=0}^{K(st)} \frac{(-st)^i}{i!} \sigma^i A \sum_{j=0}^{K(s(1-t))} \frac{(-s(1-t))^j}{j!}\sigma^j \rho \sum_{k=0}^{K(s)} \frac{(-s)^k}{k!} \sigma^k \rho}} \leq\nn
    &\hspace{20pt}\abs{\Tr{\sum_{i=K(st)+1}^{\infty} \frac{(-st)^i}{i!}\sigma^i A e^{-s(1-t)\sigma} \rho e^{-s\sigma}\rho}} + \abs{\Tr{e^{-st\sigma} A \sum_{j=0}^{K(s(1-t))} \frac{(-s(1-t))^j}{j!}\sigma^j \rho e^{-s\sigma}\rho}} \nn 
    &\hspace{20pt} +\abs{\Tr{e^{-st\sigma} A e^{-s(1-t)\sigma} \rho \sum_{k=0}^{K(s)} \frac{(-s)^k}{k!} \sigma^k\rho}} + \bigo{\epsilon_o^2} \\
    &\leq \epsilon_o \infnorm{\partial_{\alpha_k}\rho} (e^{-s(2-t)\lambda} + e^{-s(1+t)\lambda} + e^{-s\lambda})\nn 
    &\hspace{10pt}+ \epsilon_o^2 \infnorm{\partial_{\alpha_k}\rho} (e^{-s\lambda} + e^{-st\lambda} + e^{-s(1-t)\lambda}) \nn
    &\hspace{20pt}+ \epsilon_o^3 \infnorm{\partial_{\alpha_k}\rho}
\end{align}

If we plug the error term into the expression $\sum_{i=0}^{n} w_i f(\frac{L}{n}i)$, we can calculate the total error incurred

\begin{align}
    &\infnorm{\partial_{\alpha_k}\rho} \sum_{i=0}^{n} w_i \left[ \epsilon_o  (e^{-i\frac{L}{n}(2-t)\lambda} + e^{-i\frac{L}{n}(1+t)\lambda} + e^{-i\frac{L}{n}\lambda}) 
    + \epsilon_o^2  (e^{-i\frac{L}{n}\lambda} + e^{-i\frac{L}{n}t\lambda} + e^{-i\frac{L}{n}(1-t)\lambda}) 
    +\epsilon_o^3  \right] \\
    &\hspace{20pt} \leq \infnorm{\partial_{\alpha_k}\rho}\sum_{i=0}^{n} \left[ \epsilon_o  (e^{-i\frac{L}{n}(2-t)\lambda} + e^{-i\frac{L}{n}(1+t)\lambda} + e^{-i\frac{L}{n}\lambda}) 
    + \epsilon_o^2  (e^{-i\frac{L}{n}\lambda} + e^{-i\frac{L}{n}t\lambda} + e^{-i\frac{L}{n}(1-t)\lambda}) 
    +\epsilon_o^3  \right] \\
    &\hspace{20pt} =\epsilon_o  (\frac{1 - e^{-\frac{L\lambda (n+1)}{n} (2-t) }  }{1 - e^{-\frac{L\lambda}{n} (2-t) }}
    + \frac{1 - e^{-\frac{L\lambda (n+1)}{n} (1+t) }  }{1 - e^{-\frac{L\lambda}{n} (1+t) }} 
    + \frac{1 - e^{-\frac{L\lambda (n+1)}{n}}  }{1 - e^{-\frac{L\lambda}{n}}}) \nn
    &\hspace{20pt} + \epsilon_o^2  (\frac{1 - e^{-\frac{L\lambda (n+1)}{n}  }  }{1 - e^{-\frac{L\lambda}{n}  }} 
    + \frac{1 - e^{-\frac{L\lambda (n+1)}{n} t }  }{1 - e^{-\frac{L\lambda}{n} t }}
    + \frac{1 - e^{-\frac{L\lambda (n+1)}{n} (1-t) }  }{1 - e^{-\frac{L\lambda}{n} (1-t) }}) 
    + \epsilon_o^3 (n+1)
\end{align}

With the assumption that $\epsilon_o^2  (\frac{1 - e^{-\frac{L\lambda (n+1)}{n}  }  }{1 - e^{-\frac{L\lambda}{n}  }} + \frac{1 - e^{-\frac{L\lambda (n+1)}{n} t }  }{1 - e^{-\frac{L\lambda}{n} t }} + \frac{1 - e^{-\frac{L\lambda (n+1)}{n} (1-t) }  }{1 - e^{-\frac{L\lambda}{n} (1-t) }}) + \epsilon_o^3 (n+1) \leq \epsilon_o  (\frac{1 - e^{-\frac{L\lambda (n+1)}{n} (2-t) }  }{1 - e^{-\frac{L\lambda}{n} (2-t) }}
    + \frac{1 - e^{-\frac{L\lambda (n+1)}{n} (1+t) }  }{1 - e^{-\frac{L\lambda}{n} (1+t) }} 
    + \frac{1 - e^{-\frac{L\lambda (n+1)}{n}}  }{1 - e^{-\frac{L\lambda}{n}}})$, we can bound the error in a simpler manner:
    
\begin{align}
    \abs{\sum_{i=0}^{n} w_i f(\frac{L}{n}i) - \sum_{i=0}^{n} w_i \widetilde{f}(\frac{L}{n}i)} &\leq 2\infnorm{\partial_{\alpha_k}\rho}\epsilon_o  (\frac{1 - e^{-\frac{L\lambda (n+1)}{n} (2-t) }  }{1 - e^{-\frac{L\lambda}{n} (2-t) }}
    + \frac{1 - e^{-\frac{L\lambda (n+1)}{n} (1+t) }  }{1 - e^{-\frac{L\lambda}{n} (1+t) }} 
    + \frac{1 - e^{-\frac{L\lambda (n+1)}{n}}  }{1 - e^{-\frac{L\lambda}{n}}}) \\
    &\leq 6\infnorm{\partial_{\alpha_k}\rho}\epsilon_o (\frac{1 - e^{-\frac{L\lambda (n+1)}{n}}  }{1 - e^{-\frac{L\lambda}{n}}}) \\
    &\leq 6\infnorm{\partial_{\alpha_k}\rho}\epsilon_o (n+1)
\end{align}

If we upper bound this error with $\epsilon_3$, we find the upper bound on $\epsilon_0$, which in turn gives us the lower bound on $K$.

\begin{align}
    \epsilon_0 &\leq  \frac{\epsilon_3}{6 \infnorm{\partial_{\alpha_k}\rho}(n+1)}\\
    K(\alpha) &\geq \bigo{\frac{\ln \frac{6 \infnorm{\partial_{\alpha_k}\rho} (n+1)}{\epsilon_3 }}{\ln \frac{e}{\alpha}}}
\end{align}

The final approximator that approximates $I$ with error $\epsilon_1 +\epsilon_2+\epsilon_3 $ is given as

\begin{align}
     \sum_{i=0}^{n} w_i \widetilde{f}(\frac{L}{n}i) &= \sum_{i=0}^{n} w_i \left[  \Tr{ \sum_{m=0}^{K(st)} \frac{(-st)^m}{m!}\sigma^m A\sum_{j=0}^{K(s(1-t)} \frac{(-s(1-t))^j}{j!} \sigma^j \rho \sum_{k=0}^{K(s)} \frac{(-s)^k}{k!} \sigma^k \rho} \right]_{s=\frac{L}{n}i} \\
     &= \sum_{i=0}^{n} w_i \left[ \sum_{m=0}^{K(st)} \frac{(-st)^m}{m!} \sum_{j=0}^{K(s(1-t))} \frac{(-s(1-t))^j}{j!} \sum_{k=0}^{K(s)} \frac{(-s)^k}{k!} \Tr{\sigma^m A  \sigma^j \rho  \sigma^k \rho} \right]_{s=\frac{L}{n}i}
\end{align}

Let us make the assumption that $\Tr{\sigma^m A  \sigma^j \rho  \sigma^k \rho}$ can be estimated within an error $\epsilon_T$ with success probability $P_s$. The maximum obtainable error is given by

\begin{align}
     &\sum_{i=0}^{n} w_i \left[ \sum_{m=0}^{K(st)} \frac{(st)^m}{m!} \sum_{j=0}^{K(s(1-t))} \frac{(s(1-t))^j}{j!} \sum_{k=0}^{K(s)} \frac{(s)^k}{k!} \epsilon_T \right]_{s=\frac{L}{n}i} = \\
     &\hspace{25pt} \epsilon_T \sum_{i=0}^{n} w_i \left[ \frac{e^{st} \Gamma(K(st)+1, st)}{(K(st))!} \frac{e^{s(1-t)} \Gamma(K(s(1-t)+1, st)}{(K(s(1-t))!} \frac{e^{s} \Gamma(K(s)+1, st)}{(K(s))!} \right]_{s=\frac{L}{n}i} \\
     &= \epsilon_T \sum_{i=0}^{n} w_i \left[e^{3s} \frac{ \Gamma(K(st)+1, st)}{(K(st))!} \frac{ \Gamma(K(s(1-t)+1, st)}{(K(s(1-t)))!} \frac{ \Gamma(K(s)+1, st)}{(K(s))!} \right]_{s=\frac{L}{n}i} \\
     &= \epsilon_T \sum_{i=0}^{n} w_i \left[e^{3s} (1-\frac{(st)^{K(st)}}{(K(st)+1)!})(1-\frac{(s(1-t))^{K(s(1-t))}}{(K(s(1-t))+1)!})(1-\frac{(s)^{K(s)}}{(K(s)+1)!})  \right]_{s=\frac{L}{n}i}\\
     &\leq  \epsilon_T \sum_{i=0}^{n} w_i \left[e^{3s}  \right]_{s=\frac{L}{n}i}\\
     &\leq \epsilon_T \sum_{i=0}^{n} e^{3\frac{L}{n}i} \\
     &= \epsilon_T \frac{e^{3\frac{L(n+1)}{n}} - 1}{e^{3\frac{L}{n}}-1}\\
     &\leq \epsilon_T (e^{3L}+\frac{1}{e^{\frac{6L}{n}}})\\
     &\leq 2\epsilon_T e^{3L} \\
     &\leq 2\epsilon_T \left(\frac{\infnorm{\partial_{\alpha_k}\rho} \infnorm{\rho}}{2\lambda \epsilon_1}^{\frac{3}{2\lambda}} \right)
\end{align}

Setting this error to be upper bounded by $\epsilon_4$, we obtain

\begin{gather}
    \epsilon_T \leq \epsilon_4 \left(\frac{2\lambda \epsilon_1}{\infnorm{\partial_{\alpha_k}\rho} \infnorm{\rho}} \right)^{\frac{3}{2\lambda}}
\end{gather}

Now if we set $\epsilon_1 = \epsilon_2 = \epsilon_3 = \epsilon_4 = \epsilon/4$, we obtain bounds for our parameters in terms of the total error $\epsilon$.

\begin{gather}
    L \geq \frac{1}{2\lambda} \log \frac{2\infnorm{\partial_{\alpha_k}\rho} \infnorm{\rho}}{ \lambda \epsilon} \\
    n \geq \kappa_\sigma \sqrt{\frac{1}{6 \lambda\epsilon} \left(\log \frac{2\infnorm{\partial_{\alpha_k}\rho} \infnorm{\rho}}{ \lambda \epsilon} \right)^3}\\
    K(\alpha) \approx \bigo{\frac{\ln \frac{24 \infnorm{\partial_{\alpha_k}\rho} (n+1)}{\epsilon }}{\ln \frac{e}{\alpha}}} \\
    \epsilon_T \leq \frac{\epsilon}{4} \left(\frac{\lambda \epsilon}{2\infnorm{\partial_{\alpha_k}\rho} \infnorm{\rho}} \right)^{\frac{3}{2\lambda}}
\end{gather}

To get the total success probability, consider the total number of traces being evaluated:

\begin{align}
    N_{traces} &= \sum_{i=0}^{n} \left[ \bigo{\frac{\ln^3 \frac{24 \infnorm{\partial_{\alpha_k}\rho} (n+1)}{\epsilon }}{\ln{\frac{e}{st}}\ln{\frac{e}{s(1-t)}}\ln{\frac{e}{s}}}} \right]_{s=\frac{L}{n}i} \\
    &= n \frac{1}{e^3}\ln^3{\frac{25 \infnorm{\partial_{\alpha_k}\rho}n}{\epsilon}}\\
    &= \frac{\kappa_\sigma}{e^3} \sqrt{\frac{1}{6 \lambda\epsilon} \left(\log \frac{2\infnorm{\partial_{\alpha_k}\rho} \infnorm{\rho}}{ \lambda \epsilon} \right)^3} \ln^3{\frac{25 \infnorm{\partial_{\alpha_k}\rho}\kappa_\sigma \sqrt{\frac{1}{6 \lambda\epsilon} \left(\log \frac{2\infnorm{\partial_{\alpha_k}\rho} \infnorm{\rho}}{ \lambda \epsilon} \right)^3}}{\epsilon}}
\end{align}

Total failure probability is upper bounded by this number times the individual failure probability:

\begin{gather}
    P_{fail,total} \leq (1-P_s)\frac{\kappa_\sigma}{e^3} \sqrt{\frac{1}{6 \lambda\epsilon} \left(\log \frac{2\infnorm{\partial_{\alpha_k}\rho} \infnorm{\rho}}{ \lambda \epsilon} \right)^3} \ln^3{\frac{25 \infnorm{\partial_{\alpha_k}\rho}\kappa_\sigma \sqrt{\frac{1}{6 \lambda\epsilon} \left(\log \frac{2\infnorm{\partial_{\alpha_k}\rho} \infnorm{\rho}}{ \lambda \epsilon} \right)^3}}{\epsilon}}
\end{gather}

Lastly, we consider the query complexity. In our access model, we'll assume that we'll need $\bigo{j+m+k}$ queries to evaluate the expression $\Tr{\sigma^m A  \sigma^j \rho  \sigma^k \rho}$ (times the scaling with A). Thus the complexity is of the form

\begin{gather}
    \bigo{n K^4} \in \bigo{\kappa_\sigma \sqrt{\frac{1}{6 \lambda\epsilon} \left(\log \frac{2\infnorm{\partial_{\alpha_k}\rho} \infnorm{\rho}}{ \lambda \epsilon} \right)^3} \ln^4 \frac{25 \infnorm{\partial_{\alpha_k}\rho} \kappa_\sigma \sqrt{\frac{1}{6 \lambda\epsilon} \left(\log \frac{2\infnorm{\partial_{\alpha_k}\rho} \infnorm{\rho}}{ \lambda \epsilon} \right)^3}}{\epsilon }}
\end{gather}

The swap test method described in \ref{sec:swaptest} lets us calculate the polynomials of density functions and unitary operators. Substituting Eq. \ref{eq:sigder} in place of $A = \partial_{\alpha_k} \sigma$ and expanding the Hermitian operator $\widetilde{H}_k = \sum_i^{N_u}b_i V_i$ in a suitable unitary basis with $N_u$ elements adds a multiplier of $N_u$ to the query complexity. From Eq. \ref{eq:swaptesterror}, we have $\epsilon_T \leq \|b\|_1 \frac{12\pi}{M}$, $M$ being the parameter in amplitude estimation. This gives us a lower bound on M

\begin{gather}
    M \geq 48 \pi \|b\|_1 \left( \frac{2 \infnorm{\partial_{\alpha_k}\rho} \infnorm{\rho}}{\lambda \epsilon} \right)^{\frac{3}{2\lambda}}
\end{gather}

The final query complexity, utilizing the boosting scheme in Chapter \ref{sec:swaptest} is then given as 

\begin{gather}
     \bigo{\kappa_\sigma \|b\|_1 N_{AE} \left( \frac{2 \infnorm{\partial_{\alpha_k}\rho} \infnorm{\rho}}{\lambda \epsilon} \right)^{\frac{3}{2\lambda}} \sqrt{\frac{1}{6 \lambda\epsilon} \left(\log \frac{2\infnorm{\partial_{\alpha_k}\rho} \infnorm{\rho}}{ \lambda \epsilon} \right)^3} \ln^4 \frac{25 \infnorm{\partial_{\alpha_k}\rho} \kappa_\sigma \sqrt{\frac{1}{6 \lambda\epsilon} \left(\log \frac{2\infnorm{\partial_{\alpha_k}\rho} \infnorm{\rho}}{ \lambda \epsilon} \right)^3}}{\epsilon }} \\
     \in \bigotilde{\kappa_\sigma \|b\|_1 \left( \frac{2 \infnorm{\partial_{\alpha_k}\rho} \infnorm{\rho}}{\lambda \epsilon} \right)^{\frac{3}{2\lambda}} \sqrt{\frac{1}{6 \lambda\epsilon}} }
\end{gather}

Here, $N_{AE}$ should be derived from the success probability:

\begin{gather}
    P_{fail,total} \leq N_u e^{-\frac{N_{AE}}{24}}\frac{\kappa_\sigma}{e^3} \sqrt{\frac{1}{6 \lambda\epsilon} \left(\log \frac{2\infnorm{\partial_{\alpha_k}\rho} \infnorm{\rho}}{ \lambda \epsilon} \right)^3} \ln^3{\frac{25 \infnorm{\partial_{\alpha_k}\rho}\kappa_\sigma \sqrt{\frac{1}{6 \lambda\epsilon} \left(\log \frac{2\infnorm{\partial_{\alpha_k}\rho} \infnorm{\rho}}{ \lambda \epsilon} \right)^3}}{\epsilon}}
\end{gather}

Which yields, for failure probability less than $\delta$:

\begin{gather}
    N_{AE} \in \bigo{\ln{\left[\frac{1}{\delta} \frac{\kappa_\sigma}{e^3} \sqrt{\frac{1}{6 \lambda\epsilon} \left(\log \frac{2\infnorm{\partial_{\alpha_k}\rho} \infnorm{\rho}}{ \lambda \epsilon} \right)^3} \ln^3{\frac{25 \infnorm{\partial_{\alpha_k}\rho}\kappa_\sigma \sqrt{\frac{1}{6 \lambda\epsilon} \left(\log \frac{2\infnorm{\partial_{\alpha_k}\rho} \infnorm{\rho}}{ \lambda \epsilon} \right)^3}}{\epsilon}} \right]}}
\end{gather}

\section{Proof of \Cref{lem:logapprox2}}\label{sec:errorproof}

We now proceed with a proof of our technical lemma that places bounds on the derivatives of the Fourier series approximation to the logarithm and also on the sampling error incurred by estimating the result.
\begin{proof}[Proof of \Cref{lem:logapprox2}]
From \cite{wiebe2019generative}, for the approximation

\begin{align}
    \log_{KLM}\sigma :=  \sum_{k=1}^K \frac{(-1)^k}{k} \sum_{l=1}^L b_l^{(k)} (\frac{i}{2})^l \sum_{m=\ceil{l/2}-M}^{\floor{l/2}+M} (-1)^m e^{i(2m-l)\frac{\sigma\pi}{2}}\label{eq:logapprox}
\end{align}
Here $b_l^{(k)}$ are the Taylor series coefficients for $(\frac{\arcsin{x}}{\pi/2})^k$: $(\frac{\arcsin{x}}{\pi/2})^k = \sum_{l=0}^\infty b_l^{(k)} x^k$, for $x \in [-1,1]$. Important features of $b_l^{(k)}$ we will utilize are $b_l^{(k)}>0$ and $\|b^{(k)}\|_1 \leq 1$. We can then write the approximation error in the derivative as

\begin{align}
    \infnorm{\partial_{\alpha_k} (\log(1-\sigma) - \log_{KLM} (1-\sigma)} &\leq \infnorm{\partial_{\alpha_k}\sum_{k=K+1}^{\infty}\frac{(-1)^k}{k}(1-\sigma)^k} \nn
    &+ \infnorm{\partial_{\alpha_k} \sum_{k=1}^K \frac{(-1)^k}{k} \sum_{l=L+1}^\infty b_l^{(k)} \sin{\frac{\pi(1-\sigma)}{2}}^l} \nn
    &+ \infnorm{\partial_{\alpha_k} \sum_{k=1}^K \frac{(-1)^k}{k} \sum_{l=1}^L b_l^{(k)} (\frac{i}{2})^l \sum_{m \in \floor{l/2}-M \cup \ceil{l/2}+M} (-1)^m e^{i(2m-l)\frac{(1-\sigma)\pi}{2}}}.
    \label{eq:apperror}
\end{align}

We bound each of the terms in \eqref{eq:apperror} separately, starting with the first one:

\begin{align}
    \infnorm{\partial_{\alpha_k}\sum_{k=K+1}^{\infty}\frac{(-1)^k}{k}(1-\sigma)^k} &= \infnorm{\sum_{k=K+1}^{\infty}\frac{(-1)^k}{k} (\partial_{\alpha_k}(1-\sigma)^k)}\\
    &\le \sum_{k=K+1}^{\infty}\frac{1}{k} \infnorm{\partial_{\alpha_k}(1-\sigma)^k}\\
    &\le \sum_{k=K+1}^{\infty}\frac{1}{k} \infnorm{\partial_{\alpha_k} \sigma} \infnorm{(1-\sigma)^{k-1}}\\
    &\le \infnorm{\partial_{\alpha_k} \sigma} \sum_{k=K+1}^{\infty}\frac{1}{k}  \infnorm{(1-\sigma)}^{k-1}\\
    &= \infnorm{\partial_{\alpha_k} \sigma} \infnorm{(1-\sigma)}^{K} \sum_{k=0}^{\infty}\frac{1}{k+K+1}  \infnorm{(1-\sigma)}^{k}\\
    &\le \infnorm{\partial_{\alpha_k} \sigma}(1-\lambda_{min})^K \frac{1}{\lambda_{min}}
\end{align}

This error to be smaller than $\epsilon_1$ if $K$ is chosen such that
\begin{align}
    \infnorm{\partial_{\alpha_k} \sigma}(1-\lambda_{min})^K \frac{1}{\lambda_{min}} &\leq \epsilon_1\\
    (1-\lambda_{min})^K &\leq \frac{\lambda_{min} \epsilon_1}{\infnorm{\partial_{\alpha_k} \sigma}}\\
    K &\geq \frac{\ln \left( \frac{\lambda_{min} \epsilon_1}{\infnorm{\partial_{\alpha_k} \sigma}}\right) }{\ln (1-\lambda_{min})}
\end{align}

Next, we tackle the second term in~\eqref{eq:apperror}:

\begin{align}
    \infnorm{\partial_{\alpha_k} \sum_{k=1}^K \frac{(-1)^k}{k} \sum_{l=L+1}^\infty b_l^{(k)} \sin^l{\frac{\pi(1-\sigma)}{2}}} &\le  \sum_{k=1}^K \frac{1}{k} \sum_{l=L+1}^\infty b_l^{(k)} \infnorm{\partial_{\alpha_k}\sin^l{\frac{\pi(1-\sigma)}{2}}}\\
    &\le \sum_{k=1}^K \frac{1}{k} \sum_{l=L+1}^\infty b_l^{(k)} l\infnorm{\sin{\frac{\pi(1-\sigma)}{2}}}^{l-1}\infnorm{\partial_{\alpha_k} \sin{\frac{\pi(1-\sigma)}{2}}} \label{eq:apperr2}
\end{align}
Next we need to bound  $\infnorm{\partial_{\alpha_k} \sin{\frac{\pi(1-\sigma)}{2}}}$.  Using the fact that $\sin$ has a series expansion that converges uniformly on any compact interval we can therefore differentiate the infinite sum term by term and thus

\begin{align}
    \infnorm{\partial_{\alpha_k} \sin{\frac{\pi(1-\sigma)}{2}}} &= \infnorm{\partial_{\alpha_k} \sum_{j=0}^\infty \frac{(-1)^j (1-\sigma)^{2j+1}}{(2j+1)!}} \\
        &\leq \infnorm{\sum_{j=0}^\infty \frac{(-1)^j }{(2j+1)!} (\partial_{\alpha_k}(1-\sigma)^{2j+1})}\\
    &\leq \infnorm{\sum_{j=0}^\infty \frac{(-1)^j }{(2j+1)!} \sum_{n=0}^{2j} (1-\sigma)^n (-\partial_{\alpha_k} \sigma)(1-\sigma)^{2j-n}}\\
    &\leq \infnorm{\partial_{\alpha_k} \sigma} \sum_{j=0}^\infty \frac{1}{(2j)!} \infnorm{1-\sigma}^{2j}\\
    &= \infnorm{\partial_{\alpha_k} \sigma} \cosh (1-\lambda_{min})\\
    &\leq \infnorm{\partial_{\alpha_k} \sigma} \frac{1+e}{2e}\label{eq:sinderiv}
\end{align}

Using $b_l^{(k)}>0$, $\|b^{(k)}\|_1 \leq 1$ and plugging~\eqref{eq:sinderiv} into \eqref{eq:apperr2} and using the Cauchy-Schwarz inequality yields 

\begin{align}
    &\sum_{k=1}^K \frac{1}{k} \sum_{l=L+1}^\infty b_l^{(k)} l\infnorm{\sin{\frac{\pi(1-\sigma)}{2}}}^{l-1}\infnorm{\partial_{\alpha_k} \sin{\frac{\pi(1-\sigma)}{2}}} \nonumber\\
    &\qquad\leq  \infnorm{\partial_{\alpha_k} \sigma} \frac{1+e}{2e} \sum_{k=1}^K \frac{1}{k} \sum_{l=L+1}^\infty b_l^{(k)} l\infnorm{\sin{\frac{\pi(1-\sigma)}{2}}}^{l-1}\\
    &\qquad\leq \infnorm{\partial_{\alpha_k} \sigma} \frac{1+e}{2e} \sum_{k=1}^K \frac{1}{k} \sum_{l=L+1}^\infty b_l^{(k)} l (1-\lambda_{min})^{l-1}\\
    &\qquad\leq \infnorm{\partial_{\alpha_k} \sigma} \frac{1+e}{2e} \sum_{k=1}^K \frac{1}{k} \sqrt{\sum_{l=L+1}^\infty (b_{l}^{(k)})^2}\sqrt{\sum_{l=L+1}^\infty l^2 (1-\lambda_{min})^{2l-2}}\\
    &\qquad\leq \infnorm{\partial_{\alpha_k} \sigma} \frac{1+e}{2e} H_K \sqrt{\sum_{l=L+1}^\infty l^2 (1-\lambda_{min})^{2l-2}}\\
    &\qquad\leq \infnorm{\partial_{\alpha_k} \sigma} \frac{1+e}{2e} H_K \frac{L (1-\lambda_{min}^2)^{L+1}}{\lambda_{min}^3(2-\lambda_{min}^2)^{\frac{3}{2}}}
\end{align}

Solving for the sufficient L to upper bound this error with $\epsilon_2$:

\begin{align}
    L \geq -\frac{1}{\ln\frac{1}{1-\lambda_m^2}} W_{-1}\left(\frac{\ln\frac{1}{1-\lambda_m^2}}{\frac{1+e^2}{e} \infnorm{\partial_{\alpha_k} \sigma} \|a\|_1^{(k)} \frac{1-\lambda_m^2}{\lambda_m^3 (2-\lambda_m^2)^{1.5}}} \frac{-1}{\frac{12}{\epsilon}}\right).
\end{align}

Lastly, we employ a similar approach to bound the third term in~\eqref{eq:apperror}.

\begin{align}
    &\infnorm{\partial_{\alpha_k} \sum_{k=1}^K \frac{(-1)^k}{k} \sum_{l=1}^L b_l^{(k)} (\frac{i}{2})^l \sum_m (-1)^m {l\choose m} e^{i(2m-l)\frac{(1-\sigma)\pi}{2}}} \nonumber\\
    &\qquad\le  \sum_{k=1}^K \frac{1}{k} \sum_{l=1}^L b_l^{(k)} (\frac{1}{2})^l \sum_m  {l\choose m} \infnorm{\partial_{\alpha_k} e^{i(2m-l)\frac{(1-\sigma)\pi}{2}}}
\end{align}

Utilizing Duhamel's formula, $\infnorm{\partial_{\alpha_k} e^{i(2m-l)\frac{(1-\sigma)\pi}{2}}}$ can be bounded:

\begin{align}
    \infnorm{\partial_{\alpha_k} e^{i(2m-l)\frac{(1-\sigma)\pi}{2}}} &\leq \abs{(2m-l)\frac{\pi}{2}} \infnorm{\int_0^1 e^{is(2m-l)\frac{(1-\sigma)\pi}{2}} (\partial_{\alpha_k} (1-\sigma)) e^{i(1-s)(2m-l)\frac{(1-\sigma)\pi}{2}} ds}\\
    &\leq \frac{\pi}{2}\abs{2m-l} \infnorm{\int_0^1 e^{is(2m-l)\frac{(1-\sigma)\pi}{2}} (\partial_{\alpha_k} (1-\sigma)) e^{i(1-s)(2m-l)\frac{(1-\sigma)\pi}{2}}} ds\\
    &\leq \frac{\pi}{2}\abs{2m-l} \int_0^1 \infnorm{e^{is(2m-l)\frac{(1-\sigma)\pi}{2}}} \infnorm{\partial_{\alpha_k} (1-\sigma)} \infnorm{e^{i(1-s)(2m-l)\frac{(1-\sigma)\pi}{2}}} ds\\
    &\leq \frac{\pi}{2}\abs{2m-l} \int_0^1  \infnorm{\partial_{\alpha_k} \sigma}  ds\\
    &\leq \frac{\pi}{2}\abs{2m-l} \infnorm{\partial_{\alpha_k} \sigma}
\end{align}

Plugging this bound back into the sum:

\begin{align}
    &\sum_{k=1}^K \frac{1}{k} \sum_{l=1}^L b_l^{(k)} (\frac{1}{2})^l \sum_m  {l\choose m} \infnorm{\partial_{\alpha_k} e^{i(2m-l)\frac{(1-\sigma)\pi}{2}}}\nn
    &\qquad\leq \infnorm{\partial_{\alpha_k} \sigma}\sum_{k=1}^K \frac{1}{k} \sum_{l=1}^L b_l^{(k)} \frac{1}{2^l} \sum_m  {l\choose m} \frac{\pi}{2} \abs{2m-l}\\
    &\qquad= \pi \infnorm{\partial_{\alpha_k} \sigma}\sum_{k=1}^K \frac{1}{k} \sum_{l=1}^L b_l^{(k)} \frac{1}{2^l} \sum_{m=0}^{\floor{\frac{l}{2}}-M}  {l\choose m} (l-2m)\\
    &\qquad\leq \pi \infnorm{\partial_{\alpha_k} \sigma}\sum_{k=1}^K \frac{1}{k} \sum_{l=1}^L b_l^{(k)} l e^{-\frac{M^2}{L}} \\
    &\qquad\leq \pi \infnorm{\partial_{\alpha_k} \sigma}e^{-\frac{M^2}{L}} L H_K
\end{align}

Now let us tackle the statistical sampling error. This part of the error stems from using a probabilistic method to calculate the expectation value in the expression

\begin{equation}
    \mathbb{E}_S \left[\Tr{\rho e^{i s \frac{\pi}{2}m\sigma} \frac{\partial \sigma}{\partial {\alpha_k}} e^{i (1-s) \frac{\pi}{2}m\sigma}} \right]
\end{equation}

Let us define the estimator $\hat{S}_n$:

\begin{equation}
    \hat{S}_n = \frac{X_1 + X_2 + ... + X_n}{n}
\end{equation}

Here $X_i$ are i.i.d. random variables $X = \Tr{\rho e^{i S \frac{\pi}{2}m\sigma} \frac{\partial \sigma}{\partial {\alpha_k}} e^{i (1-S) \frac{\pi}{2}m\sigma}}$ with S being a uniform random variable over $[0,1]$. It can easily be seen that this estimator is unbiased:

\begin{equation}
    \exptn{\hat{S}_n} = \mathbb{E}_S \left[\Tr{\rho e^{i s \frac{\pi}{2}m\sigma} \frac{\partial \sigma}{\partial {\alpha_k}} e^{i (1-s) \frac{\pi}{2}m\sigma}} \right]
\end{equation}

The variance of the estimator is given as

\begin{equation}
    \varnc{\hat{S}_n} = \frac{1}{n} \varnc{X_i}
\end{equation}

The variance of $X_i$ can be upper bounded:

\begin{align}
    \varnc{X_i} &= \varnc{\Tr{\rho e^{i S \frac{\pi}{2}m\sigma} \frac{\partial \sigma}{\partial {\alpha_k}} e^{i (1-S) \frac{\pi}{2}m\sigma}}}\\
    &\leq \exptn{\left( \Tr{\rho e^{i s \frac{\pi}{2}m\sigma} \frac{\partial \sigma}{\partial {\alpha_k}} e^{i (1-s) \frac{\pi}{2}m\sigma}} \right)^2}\\
    &= \int_0^1 \left( \Tr{\rho e^{i s \frac{\pi}{2}m\sigma} \frac{\partial \sigma}{\partial {\alpha_k}} e^{i (1-s) \frac{\pi}{2}m\sigma}} \right)^2 ds \\
    &\leq \int_0^1 \infnorm{\frac{\partial \sigma}{\partial {\alpha_k}}}^2 ds\\
    &\leq \infnorm{\frac{\partial \sigma}{\partial {\alpha_k}}}^2
\end{align}

Thus 

\begin{equation}
    \varnc{\hat{S}_n} \leq \frac{\infnorm{\frac{\partial \sigma}{\partial {\alpha_k}}}^2}{n}
\end{equation}

Using Chebyshev's inequality and a standard statistical procedures, we repeat this sampling process $N_c$ times to ensure that the median of the repeats is close to the mean with very high probability. Chebyshev's inequality gives us for a single sampling process 

\begin{equation}
    P(\abs{\hat{S}_n - \mathbb{E}_S \left[\Tr{\rho e^{i s \frac{\pi}{2}m\sigma} \frac{\partial \sigma}{\partial {\alpha_k}} e^{i (1-s) \frac{\pi}{2}m\sigma}} \right]} \geq k\sqrt{{\rm var}[\hat S_n]}) \leq \frac{1}{k^2}
\end{equation}

Setting $k=2$, the probability of failure for a single repeat is obtained as $\frac{1}{4}$. After $N_c$ repeats, the probability that the median outcome is outside the margin defined by the Chebyshev's inequality is given as 

\begin{equation}
    p_{fail} \leq e^{-\frac{N_c}{24}}.
\end{equation}

Upper bounding $p_{fail}$ by $\delta$, the sufficient value of $N_c$ is obtained as $N_c = 24 \ln(1/\delta)$. Thus we can obtain the value of $\mathbb{E}_S \left[\Tr{\rho e^{i s \frac{\pi}{2}m\sigma} \frac{\partial \sigma}{\partial {\alpha_k}} e^{i (1-s) \frac{\pi}{2}m\sigma}} \right]$ within $\varepsilon_{s} := 2\sigma = 2\frac{\infnorm{\frac{\partial \sigma}{\partial {\alpha_k}}}}{\sqrt{n}}$ with probability of failure upper bounded by $\delta$; using $nN_c$ evaluations of $\Tr{\rho e^{i s \frac{\pi}{2}m\sigma} \frac{\partial \sigma}{\partial {\alpha_k}} e^{i (1-s) \frac{\pi}{2}m\sigma}}$.\\

The last step is to plug in the error margin into \eqref{eq:41derivative} to obtain the variance on the sample mean from this estimation process which uses $n$ samples.
Specifically, we find that 

\begin{align}
    \epsilon_{sampling} &= \sum_{m=-M}^{M} \abs{\frac{i \pi m c_m }{2}} \varepsilon_s \\
    &\leq \pi \frac{\infnorm{\frac{\partial \sigma}{\partial {\alpha_k}}}}{\sqrt{n}} \varepsilon_s M \|c\|_1 \\
    &\leq \pi \frac{\infnorm{\frac{\partial \sigma}{\partial {\alpha_k}}}}{\sqrt{n}} \varepsilon_s M \|a\|_1 \\
    &\leq \pi \frac{\infnorm{\frac{\partial \sigma}{\partial {\alpha_k}}}}{\sqrt{n}} M H_k \varepsilon_s
\end{align}

Here we have used the fact that $\|c\|_1 \leq \|a\|_1$ from Corollary 15 of \cite{wiebe2019generative}. Thus we reach at the lower bound

\begin{align}
    M &\geq \sqrt{\frac{L}{2}\ln \left( \frac{\pi H_K \infnorm{\partial_{\alpha_k} \sigma}L}{\epsilon_3} \right)}.
\end{align}

Utilizing the lower bounds in the proof, we can express $K$, $L$ and $M$ in the big-O (Bachmann-Landau) notation:

\begin{align}
    K &= \bigo{\lambda^{-1} \log(\epsilon^{-1} \lambda^{-1} \infnorm{\partial \sigma})} \\
    H_K &= \bigo{\log K}\\
    &= \bigo{\log(\frac{1}{\lambda})+ \log\log(\infnorm{\partial \sigma} \frac{1}{\lambda} \frac{1}{\epsilon})}.
\end{align}

Utilizing the properties of the Lambert-W function, the asymptotical bound for $L$ is found out to be:

\begin{align}
    L &= \bigo{\infnorm{\partial \sigma} H_K \frac{1}{\lambda} \frac{1}{\epsilon}}\\
    &= \bigo{ \left(\log(\frac{1}{\lambda})+ \log\log( \infnorm{\partial \sigma}\frac{1}{\lambda} \frac{1}{\epsilon}) \right)\infnorm{\partial \sigma} \frac{1}{\lambda} \frac{1}{\epsilon}}.
\end{align}

Lastly, we express M in big-O tilde due to the excessive number of logarithmic terms:

\begin{align}
    M &=\bigotilde{\sqrt{\infnorm{\partial \sigma} \frac{1}{\lambda} \frac{1}{\epsilon}}}.
\end{align}

\end{proof}

%\subsection{Series Expansions for Matrix Inverse}

\section{Calculation of traces of polynomials of density operators} \label{sec:swaptest}

\begin{lemma}(Generalized Swap test)
\label{lem:htest}
Expressions of the form ${\rm Tr}(\prod_i U_i \sigma_i)$ can be calculated within $\epsilon_t$ with probability $p_{succ}$ with $\mathcal{O}(\log((1-p_{succ})^{-1})/\epsilon_t)$ calls to an oracle.
\end{lemma}

\begin{proof}
A SWAP/Hadamard test like circuit can be used to calculate such quantities.  This circuit is shown in figure~\ref{fig:gen_swap}.

\begin{table}[H]
\[
\begin{array}{c}
\centering

\Qcircuit @C=1em @R=.7em {
\lstick{\op{0}{0}} & \gate{\text{Had}}&\ctrl{1}&\ctrl{2}&\qw&\cdots&&\ctrl{4}&\ctrl{1}&\gate{\text{Had}}&\qw \\
\lstick{\rho_1}    & \qw       &\gate{U_1}&\qw&\qw&\cdots&&\qw&\multigate{3}{\text{Perm}}&\qw&\qw \\
\lstick{\rho_2}    & \qw       &\qw&\gate{U_2}&\qw&\cdots&&\qw&\ghost{\text{Perm}}&\qw&\qw\\
&\vdots & & & & & \\
\lstick{\rho_n}    & \qw       &\qw&\qw&\qw&\cdots&&\gate{U_n}&\ghost{\text{Perm}}&\qw&\qw
}
\end{array}
\]
\caption{Quantum circuit for computing the expectation velue of the product of unitary matrices and density operators.  Note here that the operations are the product of unitary matrices rather than their conjugation, which would not affect the trace.  Perm here refers to a cyclic permutation of the density operators.} \label{fig:gen_swap} 
\end{table}

In this circuit, the probability of measuring $\ket{0}$ on the auxiliary qubit is given by~\cite{kieferova2021quantum}:

\begin{gather}
    {\rm Tr}((\ketbra{0}{0} \otimes \mathbb{1}) U (\ketbra{0}{0} \otimes \rho_1 \otimes \cdots \otimes \rho_n) U^\dagger) = \frac{1 + {\rm Re}[{\rm Tr}(\prod_{i=1}^n U_i \rho_i)]}{2}\label{eq:genswaptest}
\end{gather}

Consider a circuit that takes inputs in the form of $\mathcal{A}\ket{} = \ket{0}\otimes\ket{\tau_\rho}\otimes\ket{\tau_\sigma}$, where $\ket{}$ is $\ket{0}^{\otimes n}$ where n is the size of the system and $\mathcal{A}$ is a unitary preparation operator whose inverse can be calculated. From \cite{BrassardQAAE}, if this circuit produces a state $\ket{\Psi} = \ket{\Psi_1} + \ket{\Psi_2}$, where $\ket{\Psi_{1,2}}$ are unnormalized ``good (measure 0)" and ``bad (measure 1)" are components of $\ket{\Psi}$ with $\braket{\Psi_1|\Psi_1} = P(0) \equiv p$. The amplitude estimation circuit outputs a number $\tilde{p}$ such that

\begin{gather}
    \abs{\tilde{p} - p} \leq 2 \pi k \frac{\sqrt{p(1-p)}}{M} + \frac{k^2 \pi^2}{M^2}
\end{gather}
with probability greater than $p_{ae} \equiv 1- \frac{1}{2(k-1)}$. Here k and M are parameters for the amplitude estimation algorithm. However, we can boost the success probability without widening our confidence interval by repeating this experiment and taking the median result. \\
For fixed values of k and M, let the algorithm be repeated n times, with $n_s$ denoting the number of outputs in the confidence interval (number of successful trials). If $n_s > \frac{n}{2}$, it is guaranteed that the median value is in the interval, hence the experiment is successful. Let $p_{succ}$ denote the probability for the whole experiment to be successful. 

\begin{gather}
    p_{succ} = P(n_s > \frac{n}{2})
\end{gather}

However, $n_s \sim Binom(n,p_{ae})$. Using the Chernoff bound and properties of the binomial distribution

\begin{gather}
    p_{succ} \geq 1 - e^{-\frac{1}{2 p_{ae}}n(p_{ae} -\frac{1}{2})^2} \\
    p_{fail} \leq e^{-\frac{n(p_{ae} - \frac{1}{2})^2}{2 p_{ae}}} \label{eq:failurebound}
\end{gather}

Which gives us an exponential drop in the failure probability with linear investment in time. For simplicity of calculations, value of $k=3$ is used. The trace in \eqref{eq:genswaptest} can be calculated within error 

\begin{gather}
   \epsilon_T = 12 \pi  \frac{\sqrt{p(1-p)}}{M} + \frac{9 \pi^2}{M^2} \leq   \frac{6 \pi}{M} + \frac{9 \pi^2}{M^2} \leq  \frac{12 \pi}{M} \label{eq:swaptesterror}
\end{gather}

With failure probability $p_{fail} \leq e^{-\frac{n}{24}}$, with $4M+2 \in \mathcal{O}(M)$ uses of $\mathcal{A}$ or its inverse. Since $M\in O(\epsilon_t^{-1})$ and the Chernoff bound above guarantees that the probability of success is at least $p_{succ}$
\end{proof}

\bibliographystyle{unsrt}
\bibliography{references}

\begin{thebibliography}{10}

\bibitem{biamonte2017quantum}
Jacob Biamonte, Peter Wittek, Nicola Pancotti, Patrick Rebentrost, Nathan
  Wiebe, and Seth Lloyd.
\newblock Quantum machine learning.
\newblock {\em Nature}, 549(7671):195--202, 2017.

\bibitem{lloyd2016quantum}
Seth Lloyd, Silvano Garnerone, and Paolo Zanardi.
\newblock Quantum algorithms for topological and geometric analysis of data.
\newblock {\em Nature communications}, 7(1):1--7, 2016.

\bibitem{reiher2017elucidating}
Markus Reiher, Nathan Wiebe, Krysta~M Svore, Dave Wecker, and Matthias Troyer.
\newblock Elucidating reaction mechanisms on quantum computers.
\newblock {\em Proceedings of the national academy of sciences},
  114(29):7555--7560, 2017.

\bibitem{harrow2009quantum}
Aram~W Harrow, Avinatan Hassidim, and Seth Lloyd.
\newblock Quantum algorithm for linear systems of equations.
\newblock {\em Physical review letters}, 103(15):150502, 2009.

\bibitem{van_Apeldoorn_2020}
Joran van Apeldoorn, Andr{\'{a} }s Gily{\'{e}}n, Sander Gribling, and Ronald
  de~Wolf.
\newblock Quantum {SDP}-solvers: Better upper and lower bounds.
\newblock {\em Quantum}, 4:230, feb 2020.

\bibitem{amin2018quantum}
Mohammad~H Amin, Evgeny Andriyash, Jason Rolfe, Bohdan Kulchytskyy, and Roger
  Melko.
\newblock Quantum boltzmann machine.
\newblock {\em Physical Review X}, 8(2):021050, 2018.

\bibitem{kieferova2017tomography}
M{\'a}ria Kieferov{\'a} and Nathan Wiebe.
\newblock Tomography and generative training with quantum boltzmann machines.
\newblock {\em Physical Review A}, 96(6):062327, 2017.

\bibitem{Quantum_VAE}
Amir Khoshaman, Walter Vinci, Brandon Denis, Evgeny Andriyash, Hossein Sadeghi,
  and Mohammad~H Amin.
\newblock Quantum variational autoencoder.
\newblock {\em Quantum Science and Technology}, 4(1):014001, sep 2018.

\bibitem{schuld2020circuit}
Maria Schuld, Alex Bocharov, Krysta~M Svore, and Nathan Wiebe.
\newblock Circuit-centric quantum classifiers.
\newblock {\em Physical Review A}, 101(3):032308, 2020.

\bibitem{mcclean2018barren}
Jarrod~R McClean, Sergio Boixo, Vadim~N Smelyanskiy, Ryan Babbush, and Hartmut
  Neven.
\newblock Barren plateaus in quantum neural network training landscapes.
\newblock {\em Nature communications}, 9(1):1--6, 2018.

\bibitem{wang2021noise}
Samson Wang, Enrico Fontana, Marco Cerezo, Kunal Sharma, Akira Sone, Lukasz
  Cincio, and Patrick~J Coles.
\newblock Noise-induced barren plateaus in variational quantum algorithms.
\newblock {\em Nature communications}, 12(1):1--11, 2021.

\bibitem{marrero2021entanglement}
Carlos~Ortiz Marrero, M{\'a}ria Kieferov{\'a}, and Nathan Wiebe.
\newblock Entanglement-induced barren plateaus.
\newblock {\em PRX Quantum}, 2(4):040316, 2021.

\bibitem{battiti1994using}
Roberto Battiti.
\newblock Using mutual information for selecting features in supervised neural
  net learning.
\newblock {\em IEEE Transactions on neural networks}, 5(4):537--550, 1994.

\bibitem{Tishby2000TheIB}
Naftali Tishby, Fernando~C Pereira, and W.~Bialek.
\newblock The information bottleneck method.
\newblock {\em ArXiv}, physics/0004057, 2000.

\bibitem{belghazi2018mutual}
Mohamed~Ishmael Belghazi, Aristide Baratin, Sai Rajeshwar, Sherjil Ozair,
  Yoshua Bengio, Aaron Courville, and Devon Hjelm.
\newblock Mutual information neural estimation.
\newblock In {\em International conference on machine learning}, pages
  531--540. PMLR, 2018.

\bibitem{Datta_2019}
Nilanjana Datta, Christoph Hirche, and Andreas Winter.
\newblock Convexity and operational interpretation of the quantum information
  bottleneck function.
\newblock {\em 2019 IEEE International Symposium on Information Theory (ISIT)},
  Jul 2019.

\bibitem{hayashi2022efficient}
Masahito Hayashi and Yuxiang Yang.
\newblock Efficient algorithms for quantum information bottleneck.
\newblock {\em arXiv preprint arXiv:2208.10342}, 2022.

\bibitem{van2014renyi}
Tim Van~Erven and Peter Harremos.
\newblock R{\'e}nyi divergence and kullback-leibler divergence.
\newblock {\em IEEE Transactions on Information Theory}, 60(7):3797--3820,
  2014.

\bibitem{muller2013quantum}
Martin M{\"u}ller-Lennert, Fr{\'e}d{\'e}ric Dupuis, Oleg Szehr, Serge Fehr, and
  Marco Tomamichel.
\newblock On quantum r{\'e}nyi entropies: A new generalization and some
  properties.
\newblock {\em Journal of Mathematical Physics}, 54(12):122203, 2013.

\bibitem{berta2017variational}
Mario Berta, Omar Fawzi, and Marco Tomamichel.
\newblock On variational expressions for quantum relative entropies.
\newblock {\em Letters in Mathematical Physics}, 107(12):2239--2265, 2017.

\bibitem{deng2012mnist}
Li~Deng.
\newblock The mnist database of handwritten digit images for machine learning
  research [best of the web].
\newblock {\em IEEE signal processing magazine}, 29(6):141--142, 2012.

\bibitem{Grimsmo}
Arne~L. Grimsmo and Susanne Still.
\newblock Quantum predictive filtering.
\newblock {\em Phys. Rev. A}, 94:012338, Jul 2016.

\bibitem{Salek_2019}
Sina Salek, Daniela Cadamuro, Philipp Kammerlander, and Karoline Wiesner.
\newblock Quantum rate-distortion coding of relevant information.
\newblock {\em IEEE Transactions on Information Theory}, 65(4):2603–2613, Apr
  2019.

\bibitem{beer2020training}
Kerstin Beer, Dmytro Bondarenko, Terry Farrelly, Tobias~J Osborne, Robert
  Salzmann, Daniel Scheiermann, and Ramona Wolf.
\newblock Training deep quantum neural networks.
\newblock {\em Nature communications}, 11(1):1--6, 2020.

\bibitem{romero2017quantum}
Jonathan Romero, Jonathan~P Olson, and Alan Aspuru-Guzik.
\newblock Quantum autoencoders for efficient compression of quantum data.
\newblock {\em Quantum Science and Technology}, 2(4):045001, 2017.

\bibitem{kieferova2021quantum}
Maria Kieferova, Ortiz~Marrero Carlos, and Nathan Wiebe.
\newblock Quantum generative training using r$\backslash$'enyi divergences.
\newblock {\em arXiv preprint arXiv:2106.09567}, 2021.

\bibitem{spekkens2013bayesian}
M.~S. Leifer and Robert~W. Spekkens.
\newblock Towards a formulation of quantum theory as a causally neutral theory
  of bayesian inference.
\newblock {\em Physical Review A}, 88(5), November 2013.

\bibitem{van2020quantum}
Joran Van~Apeldoorn, Andr{\'a}s Gily{\'e}n, Sander Gribling, and Ronald
  de~Wolf.
\newblock Quantum sdp-solvers: Better upper and lower bounds.
\newblock {\em Quantum}, 4:230, 2020.

\bibitem{low2019hamiltonian}
Guang~Hao Low and Isaac~L Chuang.
\newblock Hamiltonian simulation by qubitization.
\newblock {\em Quantum}, 3:163, 2019.

\bibitem{measuredrenyi_cirac}
Samuel~O. Scalet, {\'{A}}lvaro~M. Alhambra, Georgios Styliaris, and J.~Ignacio
  Cirac.
\newblock Computable {R}{\'{e}}nyi mutual information: {A}rea laws and
  correlations.
\newblock {\em {Quantum}}, 5:541, September 2021.

\bibitem{measuredrenyi_tomamichel}
Mario Berta, Omar Fawzi, and Marco Tomamichel.
\newblock On variational expressions for quantum relative entropies.
\newblock {\em Letters in Mathematical Physics}, 107(12):2239--2265, sep 2017.

\bibitem{ohya2004quantum}
Masanori Ohya and D{\'e}nes Petz.
\newblock {\em Quantum entropy and its use}.
\newblock Springer Science \& Business Media, 2004.

\bibitem{Childs_2017}
Andrew~M. Childs, Robin Kothari, and Rolando~D. Somma.
\newblock Quantum algorithm for systems of linear equations with exponentially
  improved dependence on precision.
\newblock {\em SIAM Journal on Computing}, 46(6):1920–1950, Jan 2017.

\bibitem{abramowitz1964handbook}
Milton Abramowitz and Irene~A Stegun.
\newblock {\em Handbook of mathematical functions with formulas, graphs, and
  mathematical tables}, volume~55.
\newblock US Government printing office, 1964.

\bibitem{wiebe2019generative}
Nathan Wiebe and Leonard Wossnig.
\newblock Generative training of quantum boltzmann machines with hidden units.
\newblock {\em arXiv preprint arXiv:1905.09902}, 2019.

\bibitem{BrassardQAAE}
G.~Brassard, P.~Hoyer, M.~Mosca, Alain Tapp Diro~Universite de~Montreal,
  B.~U.~O. Aarhus, and Cacr~University of~Waterloo.
\newblock Quantum amplitude amplification and estimation.
\newblock {\em arXiv: Quantum Physics}, 2000.

\end{thebibliography}
\end{document}